\documentclass[pra,aps,twocolumn,superscriptaddress,showpacs,showkeys,nofootinbib]{revtex4-1}
\usepackage[latin1]{inputenc}
\usepackage{amsmath,amsthm,amssymb,graphicx,subfigure,xcolor}
\allowdisplaybreaks[4]
\usepackage{mathrsfs}
\usepackage{graphics,float}
\usepackage{epstopdf}
\usepackage{color}
\usepackage[unicode=true]{hyperref}
\usepackage{booktabs}
\usepackage{tabularx}
\usepackage{tabu}
\usepackage{multirow}
\hypersetup{
     colorlinks=true,       		
     linkcolor=red,          	
     citecolor=blue,            
     urlcolor=blue,           	
 }
\newtheorem{theorem}{Theorem}
\newtheorem*{theorem*}{Theorem}
\newtheorem{corollary}{Corollary}
\newtheorem*{corollary*}{Corollary}
\newtheorem{lemma}{Lemma}
\newtheorem*{lemma*}{Lemma}

\newtheorem*{proposition*}{Proposition}
\theoremstyle{definition}

\newtheorem*{definition*}{Definition}
\theoremstyle{remark}

\newtheorem*{remark*}{Remark}

\newtheorem*{example*}{Example}
\def\ba{\begin{array}}
\def\ea{\end{array}}
\def\be{\begin{equation}}
\def\ee{\end{equation}}
\newcommand{\norm}[1]{\left\Vert#1\right\Vert}
\newcommand{\abs}[1]{\left\vert#1\right\vert}
\newcommand{\set}[1]{\left\{#1\right\}}
\newcommand{\Real}{\mathbb R}
\newcommand{\eps}{\varepsilon}
\newcommand{\To}{\longrightarrow}
\newcommand{\BX}{\mathbf{B}(X)}
\newcommand{\A}{\mathcal{A}}
\newcommand{\proj}[1]{\ket{#1}\bra{#1}}
\newcommand{\ket}[1]{|#1\rangle}
\newcommand{\bra}[1]{\langle#1|}
\newcommand{\tr}{\rm tr}
\begin{document}
\title{Parameterized bipartite entanglement measures and entanglement constraints}
\author{Wen Zhou}
\email{2230501027@cnu.edu.cn}
\affiliation{School of Mathematical Sciences, Capital Normal University, Beijing 100048, China}
\author{Zhong-Xi Shen}
\email{18738951378@163.com}
\affiliation{School of Mathematical Sciences, Capital Normal University, Beijing 100048, China}
\author{Dong-Ping Xuan}
\email{2230501014@cnu.edu.cn}
\affiliation{School of Mathematical Sciences, Capital Normal University, Beijing 100048, China}
\author{Zhi-Xi Wang}
\email{wangzhx@cnu.edu.cn}
\affiliation{School of Mathematical Sciences, Capital Normal University, Beijing 100048, China}
\author{Shao-Ming Fei}
\email{feishm@cnu.edu.cn}
\affiliation{School of Mathematical Sciences, Capital Normal University, Beijing 100048, China}
\begin{abstract}
In this paper, we propose a novel class of parameterized entanglement measures which are named as $G_\omega$-concurrence ($G_\omega$C) ($0<\omega\leq1$), and demonstrate comprehensively that they satisfy all the necessary axiomatic conditions required for an entanglement measure. Furthermore, we derive an analytical formula relating $G_\omega$C to concurrence for the range of $0.85798\leq\omega\leq1$ within two-qubit systems. Additionally, we prove a new polygamy relation of multiqubit quantum entanglement in terms of $G_\omega$-concurrence of assistance ($G_\omega$CoA). However, it fails to obey the monogamy relation, but we have demonstrated that the squared $G_\omega$-concurrence (S$G_\omega$C) does obeys a general monogamy relation in an arbitrary $N$-qubit mixed state. Based on the monogamy properties of S$G_\omega$C, we can construct the corresponding multipartite entanglement indicators, which can detect all genuine multiqubit entangled states even in the case of $N$-tangle vanishes. In addition, for multipartite higher-dimensional systems, it is illustrated that S$G_\omega$C
still has the applicability of the monogamy relation.
\end{abstract}



\maketitle

\section{Introduction}\label{I}
Quantum entanglement stands out as an important physical resource in quantum information tasks such as quantum teleportation \cite{Bennett18951993}, quantum dense coding \cite{Bennett28811992}, quantum secret sharing \cite{Hillery18291999} and quantum cryptography \cite{Gisin1452002}. One of the key challenges within the theory of quantum entanglement is the quantification of entanglement \cite{Vedral16191998,Vedral22751997}. Some interesting entanglement measures for bipartite entangled systems have been presented, such as the concurrence \cite{Hill50221997,Rungta0423152001,Wootters22451998}, entanglement of formation \cite{Bennett38241996,Horodecki32001}, negativity \cite{Zyczkowski8831998,Vidal0323142002}, Tsallis-$q$ entropy of entanglement \cite{Kim0623282010} and R$\acute{e}$nyi-$\alpha$ entropy of entanglement \cite{Gour0121082007,Kim4453052010}.

Another key point is the study of the sharability and distribution of quantum entanglement in multipartite systems. Coffman, Kundu, and Wootters introduced first the quantitative description of the monogamy of entanglement (MoE) satisfied by the squared concurrence \cite{Coffman0523062000} in three-qubit quantum systems. Osborne and Verstraete extended the monogamy relation to the case of $N$-qubit case \cite{Osborne2205032006}. Later, extensive researches have been conducted on the distribution of entanglement in multipartite quantum systems by employing various entanglement measures such as the squared entanglement of formation \cite{Oliveira0343032014,Bai1005032014,Bai0623432014}, the squared R$\acute{e}$nyi-$\alpha$ entropy \cite{Song0223062016}, the squared Tsallis-$q$ entropy \cite{Luo0623402016} and the squared Unified-$(r,s)$ entropy \cite{Khan164192019}.

Dual to a bipartite entanglement measure is the entanglement of assistance, which generally gives rise to polygamy relations. For arbitrary tripartite systems, Gour et al. in Ref.~\cite{Gour0121082007} established the first polygamy inequality by using the squared concurrence of assistance (CoA). Similar to the monogamy relation, the corresponding polygamy relations have also been established for different kinds of assisted entanglement measures \cite{Kim0123342018,Yang5452019}.

The monogamy relations play an important role in quantum information theory \cite{Seevinck2732010}, condensed-matter physics \cite{Ma3992011} and even black-hole physics \cite{Ve1072013}. A monogamy inequality is related to a residual quantity~\cite{Zhu0243042014}. For example, the squared concurrence corresponds to the tangle. It is acknowledged that the tangle is unable to detect the entanglement of states like the W state \cite{Kim4453052010}. To compensate this deficiency, it is imperative to establish alternative monogamy relations beyond squared concurrence.

The following outlines the structure of the paper. In Sec.\ref{s1}, we introduce a class of bipartite entanglement measures
$G_\omega$C ($0<\omega\leq1$) that satisfies criteria such as faithfulness, invariance under local unitary transformations, (strong) monotonicity and convexity. In Sec.\ref{III}, an analytic relation between $G_\omega$C and concurrence is provided for $0.85798\leq\omega\leq1$ in two-qubit quantum systems. In Sec.\ref{IV} and Sec.\ref{V}, we discuss the polygamy and monogamy properties in terms of $G_\omega$CoA and S$G_\omega$C. Moreover, based on the monogamy properties of S$G_\omega$C,
we develop a set of multipartite entanglement indicators for $N$-qubit states, which are shown to work for the cases for which the usual concurrence-based indicators do not work. In addition, for multilevel systems, it is illustrated that the S$G_\omega$C can be monogamous even if the squared concurrence is polygamous. The main conclusions are summarized in Sec.\ref{VI}.

\section{$G_\omega$-Concurrence}\label{s1}

In this section, we present a category of parameterized bipartite entanglement measures $\mathfrak{C}_\omega$ known as $G_\omega$-Concurrence $G_\omega$C. For any bipartite pure state $|\varphi\rangle_{AB}$, the $G_\omega$C of $|\varphi\rangle_{AB}$ is defined as
\begin{equation}\label{2.1}
\begin{array}{rl}
\mathfrak{C}_\omega(|\varphi\rangle_{AB})=[{\rm tr}(\rho_{A}^\omega)-1]^\omega,
\end{array}
\end{equation}
where $0<\omega\leq1$, $\rho_{A}={\rm tr}_B(|\varphi\rangle_{AB}\langle\varphi|)$ is reduced
density obtained by tracing over the subsystem $B$.

For a pure state $|\varphi\rangle_{AB}$ in Schmidt decomposition form,
\begin{equation}\label{2.2}
\begin{array}{rl}
|\varphi\rangle_{AB}=\sum_{i=1}\sqrt{\chi_i}|i_A\rangle|i_B\rangle,
\end{array}
\end{equation}
where $\chi_i$ are non-negative real numbers with $\sum_i\chi_i=1$, one has
\begin{equation}\label{2.3}
\begin{array}{rl}
\mathfrak{C}_\omega(|\varphi\rangle_{AB})=(\sum_{i=1}\chi_i^\omega-1)^\omega.
\end{array}
\end{equation}

For a bipartite mixed state $\rho_{AB}$, the $G_\omega$C is given via the convex-roof extension,
\begin{equation}\label{2.4}
\begin{array}{rl}
\mathfrak{C}_{\omega}(\rho_{AB})=\min\limits_{\{p_i,|\varphi_{i}\rangle\}}\sum\limits_{i}p_i\mathfrak{C}_{\omega}(|\varphi_i\rangle),
\end{array}
\end{equation}
where the minimum is taken over all possible pure state decompositions of $\rho_{AB}=\sum_{i}p_i|\varphi_{i}\rangle \langle\varphi_{i}|$.
The $G_\omega$CoA, a dual quantity of $G_\omega$C, is given by
\begin{equation}\label{2.5}
\begin{array}{rl}
\mathfrak{C}_{\omega}^a(\rho_{AB})
=\max\limits_{\{p_i,|\varphi_{i}\rangle\}}\sum\limits_{i}
p_i\mathfrak{C}_{\omega}(|\varphi_i\rangle),
\end{array}
\end{equation}
where the maximum runs over all possible pure state decompositions of $\rho_{AB}=\sum_{i}p_i|\varphi_{i}\rangle \langle\varphi_{i}|$.
If $\rho_{AB}$ is a pure state, then $\mathfrak{C}_{\omega}^a(\rho_{AB})=\mathfrak{C}_{\omega}(\rho_{AB})$.

A bone fide entanglement measure should satisfy certain conditions \cite{Horodecki8652009}: (i) Faithfulness; (ii) Invariance under any local unitary transformation; (iii) Monotonicity under local operation and classical communication (LOCC). We show that $G_\omega$C not only meets the conditions (i)-(iii), but also the following properties, (iv) Entanglement monotone \cite{Vidal3552000} (or strong monotonicity under LOCC); (v) Convexity.

{\it i)} We prove now that $\mathfrak{C}_\omega(\rho_{AB})\geq0$ for $0<\omega\leq1$, and $\mathfrak{C}_\omega(\rho_{AB})=0$ if and only if $\rho_{AB}$ is a separable state.

It is obvious that $\mathfrak{C}_{\omega}(\rho_{AB})\geq0$ since $1\leq{\rm tr}(\rho_A^\omega)$ for $0<\omega\leq1$.

$\Leftarrow$: If a bipartite pure state $|\varphi\rangle_{AB}$ is separable, then ${\rm tr}(\rho_A^\omega)=1$, which results in $\mathfrak{C}_\omega(|\varphi\rangle_{AB})=0$.
For any separable mixed state $\rho_{AB}$, assume $\{p_i,|\varphi_i\rangle_{AB}\}$ is the separable pure state decomposition of $\rho_{AB}$. We obtain $\mathfrak{C}_\omega(\rho_{AB})\leq\sum_ip_i\mathfrak{C}_\omega(|\varphi_i\rangle_{AB})=0$. Due to $\mathfrak{C}_\omega(\rho_{AB})\geq0$, we get $\mathfrak{C}_\omega(\rho_{AB})=0$.

$\Rightarrow$: For a bipartite pure state $|\varphi\rangle_{AB}=\sum_i\sqrt{\chi_i}|i_A\rangle|i_B\rangle$,
we have the reduced density operator $\rho_A=\sum_i\chi_i|i_A\rangle\langle i_A|$.
If $\mathfrak{C}_\omega(|\varphi\rangle_{AB})=0$ for $0<\omega\leq1$, according to Eq.~(\ref{2.3}) and $0\leq\chi_i\leq1$, the Schmidt number of $|\varphi\rangle_{AB}$ must be one. Hence $|\varphi\rangle_{AB}$ is separable.
For any mixed state $\rho_{AB}$, if $\mathfrak{C}_\omega(\rho_{AB})=0$, by the definition of $G_\omega$C we have $\mathfrak{C}_\omega(|\varphi_i\rangle_{AB})=0$ for any $i$, which implies that each $|\varphi_i\rangle_{AB}$ is separable, and thus $\rho_{AB}$ is separable.

{\it ii)} Concerning that an entanglement measure should not increase under LOCC \cite{Van0605042013}, we first introduce two lemmas.

\begin{lemma}\label{l3}
\cite{Mintert2072005}
A state $|\psi\rangle$ can be derived from the state $|\varphi\rangle$ utilizing only LOCC $\Leftrightarrow$ the vector $\vec{\chi}_\psi$ majorizes $\vec{\chi}_\varphi$ ($\vec{\chi}_\varphi\prec\vec{\chi}_\psi$), where $\vec{\chi}_\varphi$ ($\vec{\chi}_\psi$) is the Schmidt vector given by the squared Schmidt coefficients of the state $|\psi\rangle$ ($|\varphi\rangle$) arrange in nonincreasing order.
\end{lemma}

\begin{lemma}\label{l4}
\cite{Ando1631989}
A function $F$ is monotone for pure states $\Leftrightarrow$ $F$ is Schur concave with respect to the spectrum of the subsystem, which is equivalent to the following two conditions:
(a) $F$ is invariant under any permutation of two arguments;
(b) For any two components $\chi_i$ and $\chi_j$ of $\vec{\chi}$, the inequality $(\chi_i-\chi_j)(\frac{\partial F}{\partial\chi_i}-\frac{\partial F}{\partial\chi_j})\leq0$ holds.
\end{lemma}

We utilize the Lemmas \ref{l3} and \ref{l4} to prove the property (ii).
For a bipartite pure state $|\varphi\rangle$, we have $\mathfrak{C}_\omega(|\varphi\rangle)=(\sum_{i=1}^{d}\chi_i^\omega-1)^{\omega}$ with $\chi_1\geq\chi_2\geq\cdots\geq\chi_d$. It is easy to get that $\mathfrak{C}_\omega(|\varphi\rangle)$ is invariant under any permutation of the arguments $\chi_i$ and $\chi_j$ in the vector $ \vec{\chi}_\varphi$, and
\begin{align*}
&(\chi_i-\chi_j)(\frac{\partial \mathfrak{C}_\omega}{\partial\chi_i}-\frac{\partial \mathfrak{C}_\omega}{\partial\chi_j})\\
=&(\chi_i-\chi_j)[\omega^2(\sum\limits_{l=1}^{m}\chi_l^\omega-1)^{\omega-1}
\chi_i^{\omega-1}\\
&
-(\sum\limits_{l=1}^{m}\chi_l^\omega-1)^{\omega-1}\chi_j^{\omega-1}]\\
=&\omega^2(\chi_i-\chi_j)(\sum\limits_{l=1}^{m}\chi_l^\omega-1)^{\omega-1}
(\chi_i^{\omega-1}-\chi_j^{\omega-1})\leq0.
\end{align*}
Therefore, $\mathfrak{C}_\omega[\Lambda_{\rm LOCC}(|\varphi\rangle)]\leq\mathfrak{C}_\omega(|\varphi\rangle)$.

For any mixed state $\rho_{AB}$, let $\{p_i,|\varphi_i\rangle\}$ be the optimal pure state decomposition of $\mathfrak{C}_\omega(\rho_{AB})$. We have
\begin{align*}
\mathfrak{C}_\omega(\rho_{AB})=\sum\limits_{i}p_i\mathfrak{C}_{\omega}(|\varphi_i\rangle)
&\geq\sum\limits_{i}p_i\mathfrak{C}_{\omega}[\Lambda_{\rm LOCC}(|\varphi_i\rangle)]\\
&\geq \mathfrak{C}_\omega[\Lambda_{\rm LOCC}(\rho_{AB})],
\end{align*}
where the last inequality is due to the definition (\ref{2.4}).

{\it iii)} The local unitary transformations are kinds of invertible LOCC operations  \cite{Horodecki8652009}. From property (ii) we immediately obtain that $\mathfrak{C}_\omega(\rho_{AB})$ is invariant under any local unitary transformation, $\mathfrak{C}_\omega(\rho_{AB})=\mathfrak{C}_\omega(U_A\otimes U_B\rho_{AB}U_A^\dagger\otimes U_B^\dagger)$.

{\it iv)} For any given states $\rho_{1}$ and $\rho_{2}$ we have
\begin{equation}
\begin{array}{rl}
{G}_\omega(\chi_{1}\rho_{1}+\chi_{2}\rho_{2})
=&[{\rm tr}(\chi_{1}\rho_{1}+\chi_{2}\rho_{2})^\omega-1]^{\omega}\\
\geq&\big\{[\chi_{1}({\rm tr}\rho_{1}^\omega)^{\frac{1}{\omega}}+\chi_{2}({\rm tr}\rho_{2}^\omega)^{\frac{1}{\omega}}]^\omega-1\big\}^\omega\\
\geq&\big[(\chi_{1}{\rm tr}\rho_{1}^\omega+\chi_{2}{\rm tr}\rho_{2}^\omega)-1\big]^{\omega}\\
=&\big[\chi_{1}({\rm tr}\rho_{1}^\omega-1)+\chi_{2}({\rm tr}\rho_{2}^\omega-1)\big]^{\omega}\\
\geq&\chi_{1}({\rm tr}\rho_{1}^\omega-1)^{\omega}+\chi_{2}({\rm tr}\rho_{2}^\omega-1)^{\omega}\\
=&\chi_{1}{G}_\omega(\rho_{1})+\chi_{2}{G}_\omega(\rho_{2}),
\end{array}
\end{equation}
where the first inequality is derived from Minkowski's inequality, $[{\rm tr}(\rho_{1}+\rho_{2})^\omega]^{\frac{1}{\omega}}\geq({\rm tr}\rho_{1}^\omega)^{\frac{1}{\omega}}+{\rm tr}\rho_{2}^\omega)^{\frac{1}{\omega}}$ for $0<\omega\leq1$, the second and third inequalities hold because of the concavity of $f(x)=x^s$ for $0<s<1$.
Therefore, ${G}_\omega(\rho)=({\rm tr}\rho^\omega-1)^\omega$
is concavity for any density operator $\rho$ for $0<\omega\leq1$.

In Ref.~\cite{Vidal3552000}, the author demonstrated that an entanglement quantifier $F$ adheres to strong monotonicity if it fulfills: (c) $m(U\rho_AU^\dagger)=m(\rho_A)$, where $m(\rho_A)=F(|\varphi\rangle_{AB})$ is concave with $\rho_A={\rm tr}_B(|\varphi\rangle\langle\varphi|)$; (d) $F$ is defined by convex roof extension for any mixed states. It is obvious that $\mathfrak{C}_{\omega}(\rho_{AB})$ meets these conditions from (iii) and the definition of $\mathfrak{C}_{\omega}(\rho_{AB})$.
Therefore, $G_\omega$C is an entanglement monotone,
$\mathfrak{C}_\omega(\rho_{AB})\geq\sum_ip_i\mathfrak{C}_\omega(\sigma_i)$,
given the ensemble $\{p_i,\rho_i\}$ obtained from the $\Lambda_{\rm LOCC}$ operation on $\rho_{AB}$.

{\it v)} Based on the definition of $\mathfrak{C}_\omega(\rho_{AB})$, $G_\omega$C is convex on quantum state $\rho_{AB}$, namely, $\mathfrak{C}_\omega(\rho_{AB})\leq\sum_ip_i\mathfrak{C}_\omega(\rho^i_{AB})$, where $\rho_{AB}=\sum_ip_i\rho^i_{AB}$ with $\sum_ip_i=1$ and $p_i>0$.

\section{$G_\omega$C and concurrence}\label{III}
The concurrence plays a crucial role in quantifying entanglement \cite{Gour2605012004}. We consider the relations between $G_\omega$C and the concurrence.
Consider a two-qubit pure state in Schmidt form, $|\varphi\rangle_{AB}=\sum_{i=1}^2\sqrt{\chi_i}|ii\rangle_{AB}$. The $G_\omega$C for this state is given by
$\mathfrak{C}_\omega(|\varphi\rangle_{AB})=(\chi_1^\omega+\chi_2^\omega-1)^\omega$.
The concurrence is given by $C(|\varphi\rangle_{AB})=2\sqrt{\chi_1\chi_2}$.
It is direct to demonstrate that
\begin{equation}\label{11}
\begin{array}{rl}
\mathfrak{C}_\omega(|\varphi\rangle_{AB})=h_\omega[C(|\varphi\rangle_{AB})],\\
\end{array}
\end{equation}
where $h_\omega(\theta)$ is given by
\begin{equation*}
\begin{array}{rl}
h_\omega(\theta)=[(\frac{1+\sqrt{1-\theta^2}}{2})^\omega
+(\frac{1-\sqrt{1-\theta^2}}{2})^\omega-1]^\omega
\end{array}
\end{equation*}
for $0\leq \theta\leq1$.

\begin{theorem}\label{t1}
For genera two-qubit mixed states $\rho_{AB}$, we have the following conclusion,
\begin{equation}\label{2.7}
\begin{array}{rl}
\mathfrak{C}_\omega(\rho_{AB})=h_\omega[C(\rho_{AB})]
\end{array}
\end{equation}
for $0.85798\leq\omega\leq1$.
\end{theorem}

\begin{proof}
We first prove that the function $h_\omega(\theta)$ is monotonically increasing with $\theta$ for $0<\omega\leq1$, namely, the following first-order derivative is nonnegative,
\begin{equation*}
\begin{array}{rl}
\frac{dh_\omega(\theta)}{d\theta}=&\frac{\omega^{2}}{2^\omega}
[(\frac{1+\sqrt{1-\theta^2}}{2})^\omega
+(\frac{1-\sqrt{1-\theta^2}}{2})^\omega-1]^{\omega-1}\\[1mm]
&\times\frac{\theta[(1-\sqrt{1-\theta^2})^{\omega-1}-(1+\sqrt{1-\theta^2})^{\omega-1}]}
{\sqrt{1-\theta^2}}.
\end{array}
\end{equation*}
Obviously, $\frac{dh_\omega(\theta)}{d\theta}\geq0$ for $0\leq\theta\leq1$ and $0<\omega\leq1$. Thus, $h_\omega(\theta)$ monotonically increases for $0\leq\theta\leq1$ as $h_\omega(\theta)$ is continuous. The equality holds only at the boundary of $\theta$.

Next we prove that $h_\omega(\theta)$ is convex with $\theta$ for $0.85798\leq\omega\leq1$, namely, the following second derivative of $h_\omega(\theta)$ is nonnegative, \begin{equation*}
\begin{array}{rl}
\mathrm{M}(\theta,\omega)\equiv\frac{d^2h_\omega(\theta)}{d\theta^2}
=\mathrm{A_{1}}\big[\mathrm{A_{2}}+\mathrm{A_{3}}(\mathrm{A_{4}}-\mathrm{A_{5}})\big],
\end{array}
\end{equation*}
where
\begin{equation*}
\begin{array}{rl}
&A_{1}=\frac{\omega^{2}}{2^\omega}[(\frac{1+\sqrt{1-\theta^2}}{2})^\omega
+(\frac{1-\sqrt{1-\theta^2}}{2})^\omega-1]^{\omega-2},\\[1mm]
&A_{2}=\frac{\omega(\omega-1)}{2^\omega}(\frac{\theta[(1-\sqrt{1-\theta^2})^{\omega-1}
-(1+\sqrt{1-\theta^2})^{\theta-1}]}{\sqrt{1-\theta^2}})^2,\\
&A_{3}=(\frac{1+\sqrt{1-\theta^2}}{2})^\omega+(\frac{1-\sqrt{1-\theta^2}}{2})^\omega-1,\\[1mm]
&A_{4}=\frac{(1-\sqrt{1-\theta^2})^{\omega-2}}{1-\theta^2}[\frac{1-\sqrt{1-\theta^2}}
{\sqrt{1-\theta^2}}+\theta^2(\omega-1)],\\[1mm]
&A_{5}=\frac{(1+\sqrt{1-\theta^2})^{\omega-2}}{1-\theta^2}[\frac{1+\sqrt{1-\theta^2}}
{\sqrt{1-\theta^2}}-\theta^2(\omega-1)].
\end{array}
\end{equation*}
We have
\begin{equation*}
\begin{array}{rl}
&\lim\limits_{\theta\rightarrow1}\mathrm{M}(\theta,\omega)\\
&=\frac{\omega^{2}}{3\cdot2^{2\omega}}\big(-1
+(\frac{1}{2})^{\omega-1}\big)\big(12\omega(\omega-1)^{3}-2(2^{\omega}-2)\\[1mm]
&~~~~(\omega-1)(\omega^{2}-5\omega+3)\big).
\end{array}
\end{equation*}

As shown in Fig.~\ref{Fig 1}, we see that there may exist a critical point $\omega_{\theta}\approx0.85798$ corresponding to
$\lim\limits_{\theta\rightarrow1}\mathrm{M}(\theta,\omega)=0$.
The second derivative is nonnegative in this region $\omega_{\theta}\leq\omega\leq1$, as shown in Fig.~\ref{Fig 2}.
\begin{figure}[htbp]
\centering
{\includegraphics[width=8cm,height=6cm]{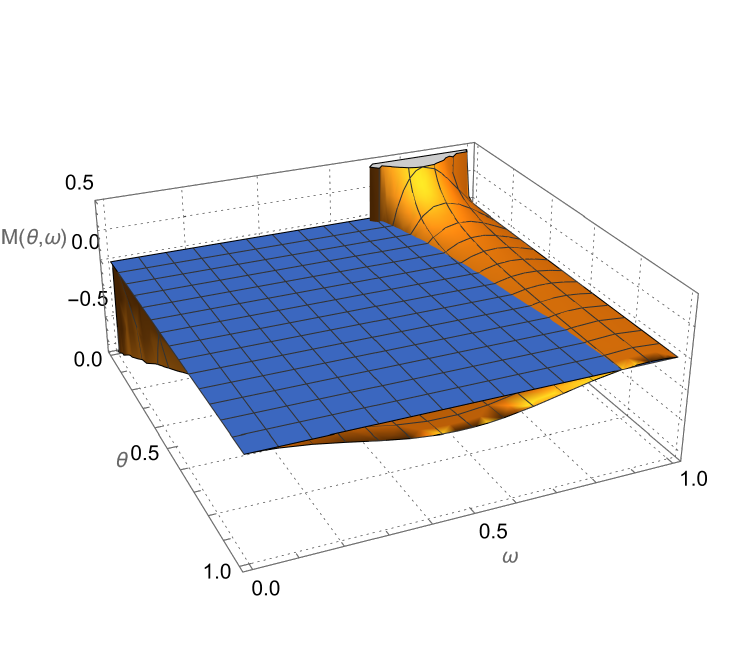}}
\caption{The orange surface represents the $\mathrm{M}(\theta,\omega)$. The blue surface represents the $\mathrm{M}(\theta,\omega)=0$.}\label{Fig 1}
\end{figure}
\begin{figure}[htbp]
\centering
{\includegraphics[width=6cm,height=4.5cm]{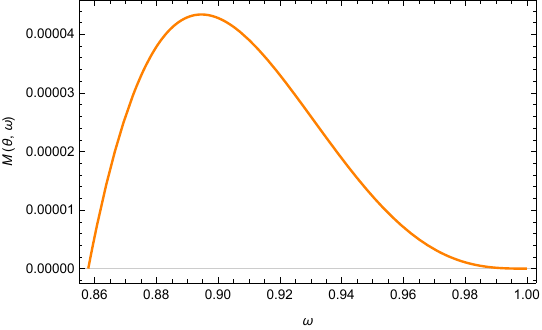}}
\caption{ The orange thick line represents the $\lim\limits_{\theta\rightarrow1}\mathrm{M}(\theta,\omega)$ for $\omega\in [\omega_{\theta},1]$.}\label{Fig 2}
\end{figure}

Based on above conclusions, we assume that $\{p_i,|\varphi_i\rangle\}$ is the optimal pure state decomposition of $\mathfrak{C}_\omega(\rho_{AB})$. We have
\begin{equation*}
\begin{array}{rl}
\mathfrak{C}_\omega(\rho_{AB})&=\sum_ip_i\mathfrak{C}_\omega(|\varphi_i\rangle)\\[1mm]
&=\sum_ip_ih_\omega[C(|\varphi_i\rangle)]\\[1mm]
&\geq h_\omega[\sum_ip_iC(|\varphi_i\rangle)]\\[1mm]
&\geq h_\omega[C(\rho_{AB})],
\end{array}
\end{equation*}
where the first inequality is due to the convexity of $h_\omega(\theta)$ for $0.8598\leq\omega\leq1$ and the second inequality follows from the monotonicity of $h_\omega(\theta)$ for $0<\omega\leq1$.

In \cite{Hill50221997} it is shown that there exists an optimal pure state decomposition $\{p_i, |\varphi_i\rangle\}$ for any
two-qubit mixed state $\rho_{AB}$ such that the concurrence of each pure state is equal. Hence,
\begin{equation*}
\begin{array}{rl}
h_\omega[C(\rho_{AB})]&=h_\omega[\sum_ip_iC(|\varphi_i\rangle)]\\[1mm]
&=\sum_ip_ih_\omega[C(|\varphi_i\rangle)]\\[1mm]
&=\sum_ip_i\mathfrak{C}_\omega(|\varphi_i\rangle)\\[1mm]
&\geq\mathfrak{C}_\omega(\rho_{AB}),
\end{array}
\end{equation*}
where the inequality holds according to the definition of $\mathfrak{C}_\omega(\rho_{AB})$.
Therefore, the Eq.~(\ref{2.7}) is true for any two-qubit mixed state.
\end{proof}

In addition to the relation between $G_\omega$C and concurrence, we have the following similar relation between $G_\omega$CoA and the coherence of assistance (CoA)  $C^a$ for two-qubit states.

\begin{corollary}\label{c1}
For any two-qubit states $\rho_{AB}$, we have
\begin{equation}\label{2.8}
\begin{array}{rl}
\mathfrak{C}_\omega^a(\rho_{AB})\leq h_\omega[C^a(\rho_{AB})]
\end{array}
\end{equation}
for $0.85798\leq\omega\leq1$.

\begin{proof}
Let $\{p_i,|\varphi_i\rangle\}$ be the pure state decomposition of $\rho_{AB}$ such that $C^a(\rho_{AB})=\sum_ip_iC(|\varphi_i\rangle)$. For $0.85798\leq\omega\leq1$, we have
\begin{equation}
\begin{array}{rl}
h_\omega[C^a(\rho_{AB})]&=h_\omega[\sum_ip_iC(|\varphi_i\rangle)]\\
&\leq\sum_ip_ih_\omega[C(|\varphi_i\rangle)]\\
&=\sum_ip_i\mathfrak{C}_\omega(|\varphi_i\rangle)\\
&\leq\mathfrak{C}_\omega^a(\rho_{AB}),
\end{array}
\end{equation}
where the first inequality is due to the convexity of $h_\omega(\theta)$ and the second inequality is assured according to the definition of $\mathfrak{C}_\omega^a(\rho_{AB})$.
\end{proof}
\end{corollary}

\section{Polygamy relation}\label{IV}
It has been shown that the squared assistance concurrence obeys the polygamy relation for an $N$-qubit state $|\varphi\rangle_{AB_1\cdots B_{N-1}}$ \cite{Gour0121082007},
\begin{equation}\label{2.9}
\begin{array}{rl}
C^2(|\varphi\rangle_{A|B_1\cdots B_{N-1}})\leq (C^a_{AB_1})^2+\cdots+(C^a_{AB_{N-1}})^2,
\end{array}
\end{equation}
where $C(|\varphi\rangle_{AB_1\cdots B_{N-1}})$ is the concurrence of $|\varphi\rangle_{A|B_1\cdots B_{N-1}}$ under bipartition $A|B_1\cdots B_{N-1}$, and $C^a_{AB_{j-1}}$ is the CoA of the reduced state $\rho_{AB_{j-1}}$, $j=2,\cdots,N$.

To investigate the polygamy property of $G_\omega$CoA defined in Eq.~(\ref{2.5}), consider
\begin{equation*}
\begin{array}{rl}
H_\omega(\theta_{1},\theta_{2})=h_\omega(\sqrt{\theta_{1}^2+\theta_{2}^2})-h_\omega(\theta_{1})-h_\omega(\theta_{2})\\
\end{array}
\end{equation*}
in the region $R=\{(\theta_{1},\theta_{2})|0\leq \theta_{1},\theta_{2},\theta_{1}^2+\theta_{2}^2\leq1\}$ for $0\leq\omega\leq1$.
Since $H_\omega(\theta_{1},\theta_{2})$ is continuous on the bounded closed set $R$, it takes maximum and minimum values for given $\omega$ at critical or boundary points. Consider
\begin{equation*}
\begin{array}{rl}
\nabla H_\omega(\theta_{1},\theta_{2})=\big(\frac{\partial H_\omega(\theta_{1},\theta_{2})}{\partial \theta_{1}},\frac{\partial H_\omega(\theta_{1},\theta_{2})}{\partial \theta_{2}}\big),
\end{array}
\end{equation*}
where
\begin{equation*}
\begin{array}{rl}
&\frac{\partial H_\omega(\theta_{1},\theta_{2})}{\partial \theta_{1}}\\
&=\frac{\theta_{1}\omega^{2}}{2^\omega}\{[(\frac{1-\sqrt{1-\theta_{1}^2-\theta_{2}^2}}{2})^\omega+(\frac{1+\sqrt{1-\theta_{1}^2-\theta_{2}^2}}{2})^\omega-1]^{\omega-1}\\
&\quad\times\frac{(1-\sqrt{1-\theta_{1}^2-\theta_{2}^2})^{\omega-1}-(1+\sqrt{1-\theta_{1}^2-\theta_{2}^2})^{\omega-1}}{\sqrt{1-\theta_{1}^2-\theta_{2}^2}}\\
&\quad-[(\frac{1-\sqrt{1-\theta_{1}^2}}{2})^\omega+(\frac{1+\sqrt{1-\theta_{1}^2}}{2})^\omega-1]^{\omega-1}\\
&\quad\times\frac{(1-\sqrt{1-\theta_{1}^2})^{\omega-1}-(1+\sqrt{1-\theta_{1}^2})^{\omega-1}}{\sqrt{1-\theta_{1}^2}}\},\\

&\frac{\partial H_\omega(\theta_{1},\theta_{2})}{\partial \theta_{2}}\\
&=\frac{\theta_{2}\omega^{2}}{2^\omega}\{[(\frac{1-\sqrt{1-\theta_{1}^2-\theta_{2}^2}}{2})^\omega+(\frac{1+\sqrt{1-\theta_{1}^2-\theta_{2}^2}}{2})^\omega-1]^{\omega-1}\\
&\quad\times\frac{(1-\sqrt{1-\theta_{1}^2-\theta_{2}^2})^{\omega-1}-(1+\sqrt{1-\theta_{1}^2-\theta_{2}^2})^{\omega-1}}{\sqrt{1-\theta_{1}^2-\theta_{2}^2}}\\
&\quad-[(\frac{1-\sqrt{1-\theta_{2}^2}}{2})^\omega+(\frac{1+\sqrt{1-\theta_{2}^2}}{2})^\omega-1]^{\omega-1}\\
&\quad\times\frac{(1-\sqrt{1-\theta_{2}^2})^{\omega-1}-(1+\sqrt{1-\theta_{2}^2})^{\omega-1}}{\sqrt{1-\theta_{2}^2}}\}.
\end{array}
\end{equation*}
Assume that $\nabla H_\omega(\theta_{1}',\theta_{2}')=0$ for some $(\theta_{1}',\theta_{2}')\in\{(\theta_{1},\theta_{2})|0\leq \theta_{1},\theta_{2},\theta_{1}^2+\theta_{2}^2\leq1\}$. Then it is observed that $\nabla H_\omega(\theta_{1}',\theta_{2}')=0$ is equivalent to $f_\omega(\theta_{1}')=f_\omega(\theta_{2}')$, where
\begin{equation*}
\begin{array}{rl}
f_\omega(n)=&[(\frac{1-\sqrt{1-n^2}}{2})^\omega+(\frac{1+\sqrt{1-n^2}}{2})^\omega-1]^{\omega-1}\\
&\times\frac{(1-\sqrt{1-n^2})^{\omega-1}-(1+\sqrt{1-n^2})^{\omega-1}}{\sqrt{1-n^2}}\}.\\
\end{array}
\end{equation*}

Denote $g_{\omega}(n)=\frac{df_\omega(n)}{dn}$. We have
\begin{equation*}
\begin{array}{rl}
g_{\omega}(n)=&\frac{n\omega(\omega-1)}{2^{\omega}}\big(\frac{(1-\sqrt{1-n^2})^{\omega-1}-(1+\sqrt{1-n^2})^{\omega-1}}{\sqrt{1-n^2}}\big)^{2}\\
&\big((\frac{(1-\sqrt{1-n^2})}{2})^{\omega}+(\frac{(1+\sqrt{1-n^2})}{2})^{\omega}-1\big)^{\omega-2}\\
&+\big(\frac{(\omega-1)n[(1-\sqrt{1-n^2})^{\omega-2}+(1+\sqrt{1-n^2})^{\omega-2}]}{1-n^2}\\
&+\frac{n[(1-\sqrt{1-n^2})^{\omega-1}-(1+\sqrt{1-n^2})^{\omega-1}]}{(1-n^2)^{\frac{3}{2}}}\big)\\
&\big((\frac{(1-\sqrt{1-n^2})}{2})^{\omega}+(\frac{(1+\sqrt{1-n^2})}{2})^{\omega}-1\big)^{\omega-1}.
\end{array}
\end{equation*}\
From Fig. \ref{fig 3} we see that $g_{\omega}(n)\leq0$ for $0<\omega\leq1$ and $0< n<1$. Namely, $f_\omega(n)$ is a strictly monotonically decreasing function with respect to $n$ for given $\omega$. Hence, $f_\omega(\theta_{1}')=f_\omega(\theta_{2}')$ implies that $\theta_{1}'=\theta_{2}'$. If $\frac{\partial H_\omega(\theta_{1},\theta_{2})}{\partial \theta_{1}}|_{(\theta_{1}',\theta_{2}')}=0$, then $f_\omega(\sqrt{2}\theta_{1}')=f_\omega(\theta_{1}')$, which contradicts to the strict monotonicity of $f_\omega(n)$. Hence $H_\omega(\theta_{1},\theta_{2})$ has no vanishing gradient in the interior of $R$.
\begin{figure}[htbp]
\centering
{\includegraphics[width=6cm,height=4cm]{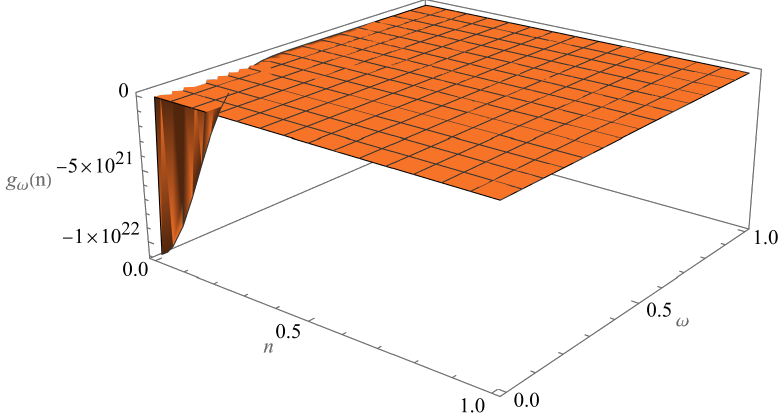}}
\caption{$g_{\omega}(n)$ as a function of $n$ and $\omega$ for $n\in(0,1)$ and $\omega\in(0,1]$.} \label{fig 3}
\end{figure}

Next we discuss the boundary values of $H_\omega(\theta_{1},\theta_{2})$ in the region $R$. If $\theta_{1}=0$ or $\theta_{2}=0$, then $H_\omega(\theta_{1},\theta_{2})=0$. If $\theta_{1}^2+\theta_{2}^2=1$, then $H_\omega(\theta_{1},\theta_{2})$ reduces to
\small\begin{equation*}
\begin{array}{rl}
p_\omega(\theta_{1})=&H_\omega(\theta_{1},\sqrt{1-\theta_{1}^{2}})\\
=&(\frac{1}{2})^{\omega^{2}}\{(2-2^{\omega})^{\omega}-[(1+\sqrt{1-\theta_{1}^2})^\omega+(1-\sqrt{1-\theta_{1}^2})^\omega\\
&-2^{\omega}]^{\omega}-[(1+\theta_{1})^\omega+(1-\theta_{1})^\omega-2^{\omega}]^{\omega}\}.\\
\end{array}
\end{equation*}\normalsize
From
\begin{equation*}
\begin{array}{rl}
\frac{dp_\omega(\theta_{1})}{d\theta_{1}}=&\frac{\omega^{2}}{2^{\omega^{2}}}\{-[-2^\omega+(1-\theta_{1})^\omega+(1+\theta_{1})^\omega]^{\omega-1}\\
&[(1+\theta_{1})^{\omega-1}-(1-\theta_{1})^{\omega-1}]\\
&-[-2^\omega+(1-\sqrt{1-\theta_{1}^2})^\omega+(1+\sqrt{1-\theta_{1}^2})^\omega]^{\omega-1}\\
&\frac{\theta[(1-\sqrt{1-\theta_{1}^2})^{\omega-1}-(1+\sqrt{1-\theta_{1}^2})^{\omega-1}]}{\sqrt{1-\theta_{1}^2}}\},
\end{array}
\end{equation*}
we have $\frac{dp_\omega(\theta_{1})}{d\theta_{1}}=0$ for $\theta_{1}=\frac{1}{\sqrt{2}}$. Since $p_\omega(0)=p_\omega(1)=0$, the sign of $p_\omega(\theta_{1})$ is determined by $p_\omega(\frac{1}{\sqrt{2}})=(\frac{1}{2})^{\omega^{2}}\{(2-2^{\omega})^{\omega}-2[(1+\sqrt{\frac{1}{2}})^\omega+(1-\sqrt{\frac{1}{2}})^\omega-2^{\omega}\}$. Fig. \ref{fig 4} shows that $p_\omega(\frac{1}{\sqrt{2}})$ is always non-positive for $0\leq\omega\leq1$.
Therefore, we have
\begin{equation}\label{2.10}
\begin{array}{rl}
h_\omega(\sqrt{\theta_{1}^2+\theta_{2}^2})\leq h_\omega(\theta_{1})+h_\omega(\theta_{2})
\end{array}
\end{equation}
for $0\leq\omega\leq1$ and $(\theta_{1},\theta_{2})\in R$.
\begin{figure}[htbp]
\centering
{\includegraphics[width=6cm,height=4cm]{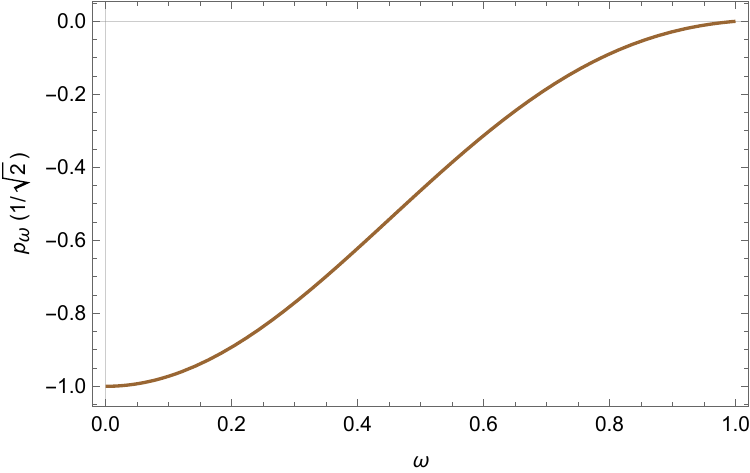}}
\caption{The function $p_\omega(1/\sqrt{2})$ is non-positive for $0< \omega\leq1$.} \label{fig 4}
\end{figure}

From (\ref{2.10}), we have the following polygamy relation satisfied by $G_\omega$CoA.

\begin{theorem}
For any $N$-qubit quantum state $\rho_{A\cdots B_{N-1}}$, we have
\begin{equation}\label{2.11}
\begin{array}{rl}
\mathfrak{C}_\omega^a(\rho_{A|B_1\cdots B_{N-1}})\leq\mathfrak{C}_\omega^a(\rho_{AB_1})+\cdots+\mathfrak{C}_\omega^a(\rho_{AB_{N-1}})
\end{array}
\end{equation}
for $0.85798\leq\omega\leq1$, where $\mathfrak{C}_\omega^a(\rho_{A|B_1\cdots B_{N-1}})$ denotes the $G_\omega$CoA of $\rho_{AB_1\cdots B_{N-1}}$ in the partition $A|B_1\cdots B_{N-1}$, $\rho_{AB_{j-1}}$ is the reduced density matrix with respect to the subsystem $AB_{j-1}$, $j=2,3,\cdots,N$.

\begin{proof} 1) $(C_{AB_1}^a)^2+\cdots+(C_{AB_{N-1}}^a)^2\leq1$. We have
\begin{equation}\label{2.12}
\begin{array}{rl}
&\mathfrak{C}_\omega(|\varphi\rangle_{A|B_1\cdots B_{N-1}})\\
=&h_\omega[C(|\varphi\rangle_{A|B_1\cdots B_{N-1}})]\\
\leq&h_\omega\big[\sqrt{(C^a_{AB_1})^2+\cdots+(C^a_{AB_{N-1}})^2}\big]\\
\leq&h_\omega(C^a_{AB_1})+h_\omega\big[\sqrt{(C^a_{AB_2})^2+\cdots+(C^a_{AB_{N-1}})^2}\big]\\
\leq&\cdots\\
\leq& h_\omega(C^a_{AB_1})+\cdots+h_\omega(C^a_{AB_{N-1}})\\
\leq&\mathfrak{C}_\omega^a(\rho_{AB_1})+\cdots+\mathfrak{C}_\omega^a(\rho_{AB_{N-1}}),
\end{array}
\end{equation}
where the first inequality holds according to the formula (\ref{2.9}) and that $h_\omega(\theta)$ is monotonically increasing for $0\leq\theta\leq1$ as shown in the proof of Theorem 1, the second and the penultimate inequalities are from (\ref{2.10}), and the last inequality is from the Corollary \ref{c1}.

2) $(C_{AB_1}^a)^2+\cdots+(C_{AB_{N-1}}^a)^2>1$. In this case there is some $j$ such that $(C_{AB_1}^a)^2+\cdots+(C_{AB_{j-1}}^a)^2\leq1$, whereas $(C_{AB_1}^a)^2+\cdots+(C_{AB_{j}}^a)^2>1$, $2\leq j\leq N$. Set
$S=(C_{AB_1}^a)^2+\cdots+(C_{AB_{j}}^a)^2-1$. Then one has
\begin{align}\label{2.13}
&\mathfrak{C}_\omega(|\varphi\rangle_{A|B_1\cdots B_{N-1}}) \notag\\
=&h_\omega[C(|\varphi\rangle_{A|B_1\cdots B_{N-1}})] \notag\\
\leq& h_\omega(1) \notag\\
=&h_\omega\big[\sqrt{(C_{AB_1}^a)^2+\cdots+(C_{AB_{j}}^a)^2-S}\big] \notag\\
\leq& h_\omega\big[\sqrt{(C_{AB_1}^a)^2+\cdots+(C_{AB_{j-1}}^a)^2}\big]\\
&+h_\omega[\sqrt{(C_{AB_{j}}^a)^2-S}] \notag\\
\leq& h_\omega(C_{AB_1}^a)+\cdots+h_\omega(C_{AB_{j-1}}^a)+h_\omega(C_{AB_{j}}^a) \notag\\
\leq&\mathfrak{C}^a_\omega(\rho_{AB_1})+\cdots+\mathfrak{C}^a_\omega(\rho_{AB_{N-1}}),\notag
\end{align}
similar to the proof of the inequality (\ref{2.12}).

Let $\rho_{AB_1\cdots B_{N-1}}$ be a multiqubit mixed state and $\{p_i,|\varphi_i\rangle_{A|B_1\cdots B_{N-1}}\}$ be the pure state decomposition such that $\mathfrak{C}_\omega^a(\rho_{A|B_1\cdots B_{N-1}})=\sum_ip_i\mathfrak{C}_\omega(|\varphi_i\rangle_{A|B_1\cdots B_{N-1}})$. We have
\begin{equation}\label{2.14}
\begin{array}{rl}
&\mathfrak{C}_\omega^a(\rho_{A|B_1\cdots B_{N-1}})\\
=&\sum\limits_ip_i\mathfrak{C}_\omega(|\varphi_i\rangle_{A|B_1\cdots B_{N-1}})\\
\leq&\sum\limits_ip_i[\mathfrak{C}^a_\omega(\rho^i_{AB_1})+\cdots+\mathfrak{C}^a_\omega(\rho^i_{AB_{N-1}})]\\
=&\sum\limits_ip_i\mathfrak{C}^a_\omega(\rho^i_{AB_1})+\cdots+\sum\limits_ip_i\mathfrak{C}^a_\omega(\rho^i_{AB_{N-1}})\\
\leq&\mathfrak{C}^a_\omega(\rho_{AB_1})+\cdots+\mathfrak{C}^a_\omega(\rho_{AB_{N-1}}),
\end{array}
\end{equation}
where $\rho^i_{AB_{j-1}}$ is the reduced density matrix of $|\varphi_i\rangle_{AB_1\cdots B_{N-1}}$ with respect to $AB_{j-1}$ system, the first inequality is derived by inequalities (\ref{2.12}) and (\ref{2.13}), and the second inequality follows from the definition of $G_\omega$-CoA.

Combining inequalities (\ref{2.12}), (\ref{2.13}), and (\ref{2.14}), we get inequality (\ref{2.11}) for any $N$-qubit quantum states.
\end{proof}
\end{theorem}

\begin{corollary}
For any $N$-qubit quantum state $\rho_{A_1\cdots A_n}$, the $\alpha$-th $(0\leq\alpha\leq1)$ power of $G_\omega$CoA fulfills the following polygamy relation,
\begin{equation}\label{poly2}
\begin{array}{rl}
[\mathfrak{C}_\omega^a(\rho_{A_1|A_2\cdots A_n})]^{\alpha}\leq[\mathfrak{C}_\omega^a(\rho_{A_1A_2})]^{\alpha}+\cdots+[\mathfrak{C}_\omega^a(\rho_{A_1A_n})]^{\alpha}.\\
\end{array}
\end{equation}

\begin{proof}Based on the relation (\ref{2.12}), one derives
\begin{equation*}
\begin{array}{rl}
&[\mathfrak{C}_\omega^a(\rho_{A_1|A_2\cdots A_n})]^{\alpha}\\
&~~\leq[\mathfrak{C}_\omega^a(\rho_{A_1A_2})+\cdots+\mathfrak{C}_\omega^a(\rho_{A_1A_n})]^{\alpha}\\
&~~\leq[\mathfrak{C}_\omega^a(\rho_{A_1A_2})]^{\alpha}+\cdots+[\mathfrak{C}_\omega^a(\rho_{A_1A_n})]^{\alpha}.\\
\end{array}
\end{equation*}
The last inequality is valid since the relation $(\sum_i\theta_i)^{\alpha}\leq\sum_i\theta_i^\alpha$ holds for $0\leq \theta_i\leq1$ and $0\leq \alpha\leq1$.
\end{proof}
\end{corollary}

Next we consider the following example to illustrate the validity of our polygamy relations of multiqubit entanglement.

{\noindent\bf Example 1}.
Let us consider the three-qubit generlized $\mathrm{W}$-class state,
\begin{equation}
|\psi\rangle_{A|BC}=\frac{1}{2}(|100\rangle+|010\rangle)+\frac{\sqrt{2}}{2}|001\rangle.
\end{equation}
We have $C_{A|BC}^{a}=\frac{\sqrt{3}}{2}$, $C_{AB}^{a}=\frac{1}{2}$ and
$C_{AC}^{a}=\frac{\sqrt{2}}{2}$. Thus we get
\begin{equation}
\begin{array}{rl}
Z_{1}&\equiv[\mathfrak{C}_\omega^a(\rho_{A|BC})]^{\alpha}
=h_\omega^{\alpha}(C_{A|BC}^{a})\\
&=[(\frac{1+\sqrt{1-(\frac{\sqrt{3}}{2})^2}}{2})^\omega+(\frac{1-\sqrt{1-(\frac{\sqrt{3}}{2})^2}}{2})^\omega-1]^{\alpha\omega},
\end{array}
\end{equation}
and
\begin{equation}
\begin{array}{rl}
Z_{2}&\equiv[\mathfrak{C}_\omega^a(\rho_{AB})]^{\alpha}+[\mathfrak{C}_\omega^a(\rho_{AC})]^{\alpha}
=h_\omega^{\alpha}(C_{AB}^{a})+h_\omega^{\alpha}(C_{AC}^{a})\\
&=[(\frac{1+\sqrt{1-(\frac{1}{2})^2}}{2})^\omega+(\frac{1-\sqrt{1-(\frac{1}{2})^2}}{2})^\omega-1]^{\alpha\omega}\\
&~~+[(\frac{1+\sqrt{1-(\frac{\sqrt{2}}{2})^2}}{2})^\omega+(\frac{1-\sqrt{1-(\frac{\sqrt{2}}{2})^2}}{2})^\omega-1]^{\alpha\omega}.
\end{array}
\end{equation}
It is seen that $G_\omega$CoA fulfills the polygamy relation, see Fig. \ref{fig 5}.  Fig. \ref{fig 6} shows the difference $Z=Z_{2}-Z_{1}$ for the $G_\omega$CoA. Fig. \ref{fig 7} shows the case of $\omega=0.9$.
\begin{figure}[htbp]
\centering
{\includegraphics[width=6cm,height=3.5cm]{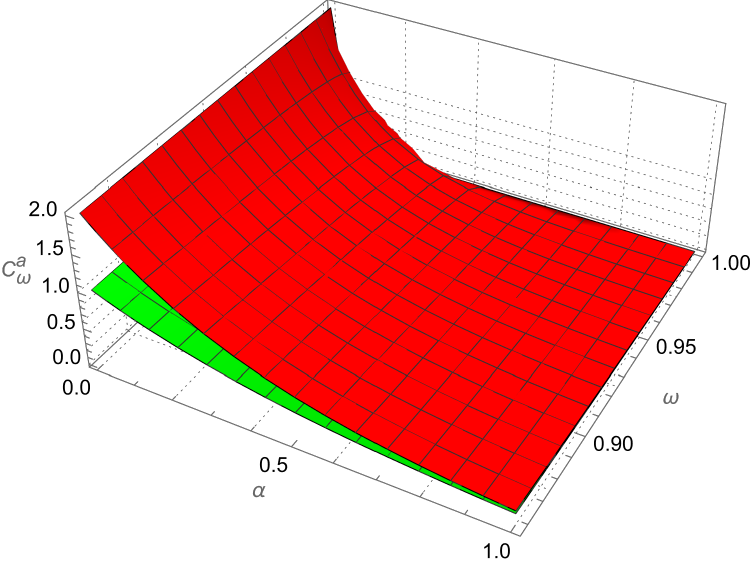}}
\caption{The green (red) surface represents the result of $Z_{2}$ ($Z_{1}$).} \label{fig 5}
\end{figure}
\begin{figure}[htbp]
\centering
{\includegraphics[width=6cm,height=3cm]{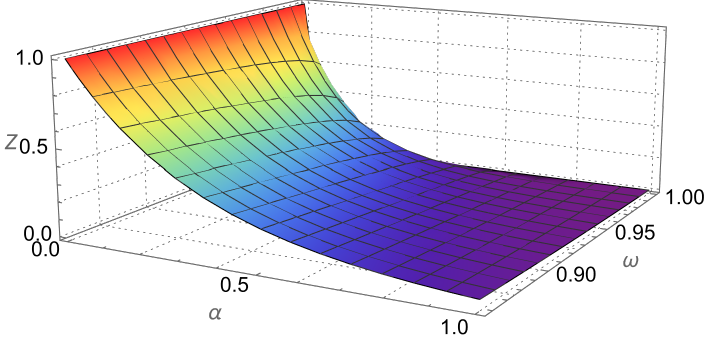}}
\caption{The surface for $Z=Z_{2}-Z_{1}$.} \label{fig 6}
\end{figure}
\begin{figure}[htbp]
\centering
{\includegraphics[width=6cm,height=4cm]{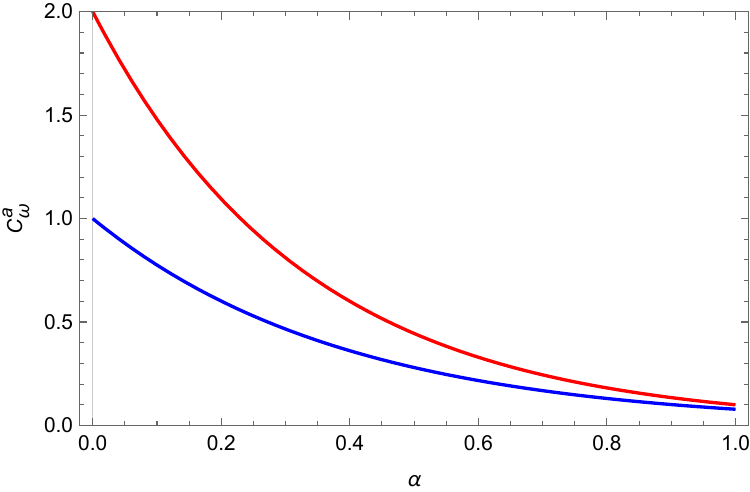}}
\caption{The red (bule) line represents $Z_{2}$ $(Z_{1})$ with $\omega=0.9$.} \label{fig 7}
\end{figure}

\section{Monogamy property}\label{V}
\subsection{Monogamy relation}
It is well known that the squared concurrence (SC) satisfies the following monogamy relation \cite{Coffman0523062000,Osborne2205032006},
\begin{equation}\label{2.15}
\begin{array}{rl}
C^2(|\varphi\rangle_{A|B_1\cdots B_{N-1}})\geq C^2_{AB_1}+\cdots+C^2_{AB_{N-1}}
\end{array}
\end{equation}
for any $N$-qubit state $|\varphi\rangle_{A\cdots B_{N-1}}$, where $C_{AB_{j-1}}$ is the concurrence of the reduced density operator $\rho_{AB_{j-1}}$, $j=2,\cdots,N$. Next we consider the monogamy property of $G_\omega$C.

Consider $\widetilde{H}_\omega(\theta_{1},\theta_{2})=h_\omega^2(\sqrt{\theta_{1}^2+\theta_{2}^2})-h_\omega^2(\theta_{1})-h_\omega^2(\theta_{2})$. From a similar analysis on ${H}_\omega(\theta_{1},\theta_{2})$ in the previous section, we have that $\nabla \widetilde{H}_\omega(\theta_{1},\theta_{2})=\big(\frac{\partial \widetilde{H}_\omega(\theta_{1},\theta_{2})}{\partial \theta_{1}},\frac{\partial \widetilde{H}_\omega(\theta_{1},\theta_{2})}{\partial \theta_{2}}\big)$ does not vanish in the interior of $R$ for $0<\omega\leq1$. If $\theta_{1}=0$ or $\theta_{2}=0$, then $\widetilde{H}_\omega(\theta_{1},\theta_{2})=0$. If $\theta_{1}^2+\theta_{2}^2=1$, set
\begin{equation*}
\begin{array}{rl}
&q_\omega(\theta_{1})=\widetilde{H}_\omega(\theta_{1},\sqrt{1-\theta_{1}^2})\\
&=(\frac{1}{2})^{2\omega^{2}}\{(2-2^\omega)^{2\omega}-[(1+\sqrt{1-\theta_{1}^2})^\omega+(1-\\
&~~\sqrt{1-\theta_{1}^2})^\omega-2^\omega]^{2\omega}-[(1+\theta_{1})^\omega+(1-\theta_{1})^\omega-2^{\omega}]^{2\omega}\}.\\
\end{array}
\end{equation*}
We have $\frac{dq_\omega(\theta_{1})}{d\theta_{1}}=0$ when $\theta_{1}=\frac{1}{\sqrt2}$ for $0<\omega\leq1$. Due to $q_\omega(0)=q_\omega(1)=0$, $q_\omega(\frac{1}{\sqrt{2}})$ determines the sign of function $q_\omega(\theta_{1})$. We plot $q_\omega(\frac{1}{\sqrt{2}})$ in Fig. \ref{fig 8}. The intermediate value theorem tells us that if a continuous function on the domain has two values with opposite signs, there must exist a root on the domain. Because $q_\omega(\frac{1}{\sqrt{2}})$ is a continuous function, and $q_{0.6}(\frac{1}{\sqrt{2}})<0$, $q_{0.8}(\frac{1}{\sqrt{2}})>0$, there is a root in the interval [0.6,~0.8]. After some straightforward calculations, we derive that the root is approximately equal to 0.7962.
We have that $q_\omega(\frac{1}{\sqrt{2}})\geq0$ for $0.7962\leq\omega\leq1$, see Fig.~\ref{fig 9}. Therefore, we have
\begin{equation}\label{2.16}
\begin{array}{rl}
h_\omega^2(\sqrt{\theta_{1}^2+\theta_{2}^2})\geq h_\omega^2(\theta_{1})+h_\omega^2(\theta_{2})
\end{array}
\end{equation}
for $0.7962\leq\omega\leq1$.
\begin{figure}[htbp]
\centering
{\includegraphics[width=6cm,height=4cm]{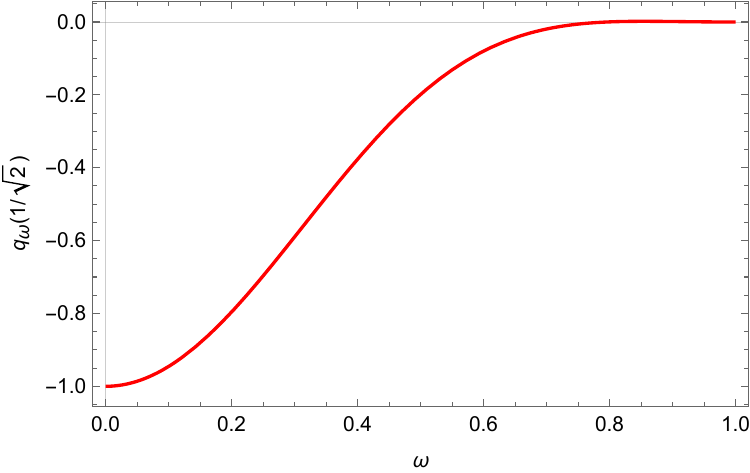}}
\caption{The function $q_\omega(\frac{1}{\sqrt{2}})$  for $0<\omega\leq1$.} \label{fig 8}
\end{figure}
\begin{figure}[htbp]
\centering
{\includegraphics[width=6cm,height=4cm]{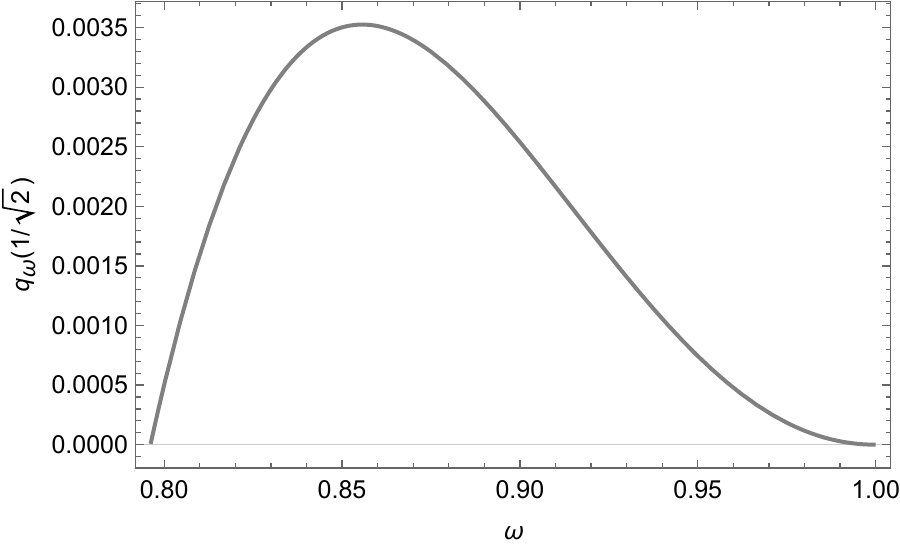}}
\caption{The function $q_\omega(\frac{1}{\sqrt{2}})$ is positive for $0.7962\leq\omega\leq1$.} \label{fig 9}
\end{figure}

Moreover, from the property of the function $h_\omega(\theta)$ introduced in the last section, we have that the function $h_\omega^2(\theta)$ is monotonically increasing with respect to $\theta$ for $0\leq \theta\leq1$ and $0<\omega\leq1$. We are ready to prove the following theorem.

\begin{theorem}\label{t3}
For an $N$-qubit quantum state $\rho_{AB_1\cdots B_{N-1}}$, we have
\begin{equation}\label{mono1}
\begin{array}{rl}
\mathfrak{C}_\omega^2(\rho_{A|B_1\cdots B_{N-1}})\geq\mathfrak{C}_\omega^2(\rho_{AB_1})
+\cdots+\mathfrak{C}_\omega^2(\rho_{AB_{N-1}})
\end{array}
\end{equation}
for $0.85798\leq\omega\leq1$.

\begin{proof}
For any $N$-qubit pure state $|\varphi\rangle_{AB_1\cdots B_{N-1}}$ and $0.85798\leq\omega\leq1$, we have
\begin{equation}\label{2.17}
\begin{array}{rl}
\mathfrak{C}_\omega^2(|\varphi\rangle_{A|B_1\cdots B_{N-1}})=&h_\omega^2[C(|\varphi\rangle_{A|B_1\cdots B_{N-1}})]\\
\geq&h_\omega^2\big(\sqrt{C_{AB_1}^2+\cdots+C_{AB_{N-1}}^2}\big)\\
\geq&h_\omega^2(C_{AB_1})\\
&+h_\omega^2\big(\sqrt{C_{AB_2}^2+\cdots+C_{AB_{N-1}}^2}\big)\\
\geq&h_\omega^2(C_{AB_1})+\cdots+h_\omega^2(C_{AB_{N-1}})\\
=&\mathfrak{C}_\omega^2(\rho_{AB_1})+\cdots+\mathfrak{C}_\omega^2(\rho_{AB_{N-1}}),
\end{array}
\end{equation}
where the first inequality is due to the fact that $h_\omega^2(\theta)$ is monotonically increasing with respect to $\theta$, the second and third inequalities can be gained by iterating formula (\ref{2.16}), the last equality follows from Theorem \ref{t1}.

Given a multiqubit mixed state $\rho_{A|B_1\cdots B_{N-1}}$. Suppose that $\{p_i,|\varphi_i\rangle_{A|B_1\cdots B_{N-1}}\}$ is the optimal pure state decomposition of $\mathfrak{C}_\omega(\rho_{A|B_1\cdots B_{N-1}})$, $\mathfrak{C}_\omega(\rho_{A|B_1\cdots B_{N-1}})=\sum_ip_i\mathfrak{C}_\omega(|\varphi_i\rangle_{A|B_1\cdots B_{N-1}})$. We have
\begin{equation*}
\begin{array}{rl}
\mathfrak{C}_\omega^2(\rho_{A|B_1\cdots B_{N-1}})=&[\sum_ip_i\mathfrak{C}_\omega(|\varphi_i\rangle_{A|B_1\cdots B_{N-1}})]^2\\
=&\{\sum_ip_ih_\omega[C(|\varphi_i\rangle_{A|B_1\cdots B_{N-1}})]\}^2\\
\geq&\{h_\omega[\sum_ip_iC(|\varphi_i\rangle_{A|B_1\cdots B_{N-1}})]\}^2\\
\geq&h_\omega^2[C(\rho_{A|B_1\cdots B_{N-1}})]\\
\geq&h_\omega^2(\sqrt{C_{AB_1}^2+\cdots+C_{AB_{N-1}}^2})\\
\geq&\mathfrak{C}_\omega^2(\rho_{AB_1})+\cdots+\mathfrak{C}_\omega^2(\rho_{AB_{N-1}}),\\
\end{array}
\end{equation*}
where the first inequality is assured owing to the convexity of $h_\omega(\theta)$ and the monotonicity of the function $y=\theta^2$ for $\theta>0$, the second inequality is based on the fact that $h_\omega^2(\theta)$ is monotonically increasing and the definition of concurrence, the third inequality is valid because the SC satisfies monogamy relation for multiqubit quantum states \cite{Osborne2205032006}, the last inequality is obtained similar to the inequality (\ref{2.17}).
\end{proof}
\end{theorem}

\begin{corollary} For any $N$-qubit quantum state $\rho_{A\cdots B_{N-1}}$, the $\beta$-th ($\beta\geq2$) power of $G_\omega$-concurrence fulfills the following monogamy relation,
\begin{equation*}
\begin{array}{rl}
\mathfrak{C}_\omega^\beta(\rho_{A|B_1\cdots B_{N-1}})\geq\mathfrak{C}_\omega^\beta(\rho_{AB_1})+\cdots+\mathfrak{C}_\omega^\beta(\rho_{AB_{N-1}}).\\
\end{array}
\end{equation*}

\begin{proof} Based on the relation (\ref{mono1}), one derives
\begin{equation*}
\begin{array}{rl}
&\mathfrak{C}_\omega^\beta(\rho_{A|B_1\cdots B_{N-1}})\\
&~~\geq[\mathfrak{C}_\omega^2(\rho_{AB_1})+\cdots+\mathfrak{C}_\omega^2(\rho_{AB_{N-1}})]^{\frac{\beta}{2}}\\
&~~\geq\mathfrak{C}_\omega^\beta(\rho_{AB_1})+\cdots+\mathfrak{C}_\omega^\beta(\rho_{AB_{N-1}}).\\
\end{array}
\end{equation*}
The last inequality is valid since the relation $(\sum_i\theta_i^2)^{\frac{\beta}{2}}\geq\sum_i\theta_i^\beta$ holds for $0\leq \theta_i\leq1$ and $\beta\geq2$.
\end{proof}
\end{corollary}

To illustrate the validity of the monogamy inequalities, we consider the following example.

{\noindent\bf Example 2}. Let us consider the three-qubit state $|\phi\rangle_{ABC}$ in
the generalized Schmidt decomposition,
\begin{equation}
|\phi\rangle_{A|BC}=\lambda_0|000\rangle+\lambda_1e^{i{\varphi}}|100\rangle+\lambda_2|101\rangle
+\lambda_3|110\rangle+\lambda_4|111\rangle,
\end{equation}
where $\lambda_i\geq0$, $i=0,1,2,3,4$, and $\sum_{i=0}^{4}\lambda_i^2 =1$. Set $\lambda_0=\lambda_3=\lambda_4=\frac{1}{\sqrt{5}}$, $\lambda_2=\sqrt{\frac{2}{5}}$ and $\lambda_1=0$. We get $C_{A|BC}=\frac{4}{5}$, $C_{AB}=\frac{2\sqrt{2}}{5}$ and
$C_{AC}=\frac{2}{5}$. Thus we have

\begin{equation}
\begin{array}{rl}
&W_{1}\equiv\mathfrak{C}_\omega^\beta(\rho_{A|BC})
=h_\omega^{\beta}(C_{A|BC})\\
&=[(\frac{1+\sqrt{1-(\frac{4}{5})^2}}{2})^\omega+(\frac{1-\sqrt{1-(\frac{4}{5})^2}}{2})^\omega-1]^{\beta\omega}
\end{array}
\end{equation}
and
\begin{equation}
\begin{array}{rl}
&W_{2}\equiv\mathfrak{C}_\omega^\beta(\rho_{AB})+\mathfrak{C}_\omega^\beta(\rho_{AC})
=h_\omega^\beta(C_{AB})+h_\omega^\beta(C_{AC})\\
&=[(\frac{1+\sqrt{1-(\frac{2\sqrt{2}}{5})^{2}}}{2})^{\beta\omega}+(\frac{1-\sqrt{1-(\frac{2\sqrt{2}}{5})^2}}{2})^\omega-1]^{\beta\omega}\\
&+[(\frac{1+\sqrt{1-(\frac{2}{5})^2}}{2})^\omega+(\frac{1-\sqrt{1-(\frac{2}{5})^2}}{2})^\omega-1]^{\beta\omega}.
\end{array}
\end{equation}
It is seen that $G_\omega$C fulfills the monogamy relation, see Fig. \ref{fig 10}. Fig. \ref{fig 11} shows the difference $W=W_{1}-W_{2}$ of the $G_\omega$C. Fig. \ref{fig 12} shows the case of $\omega=0.9$.
\begin{figure}[htbp]
\centering
{\includegraphics[width=7cm,height=5cm]{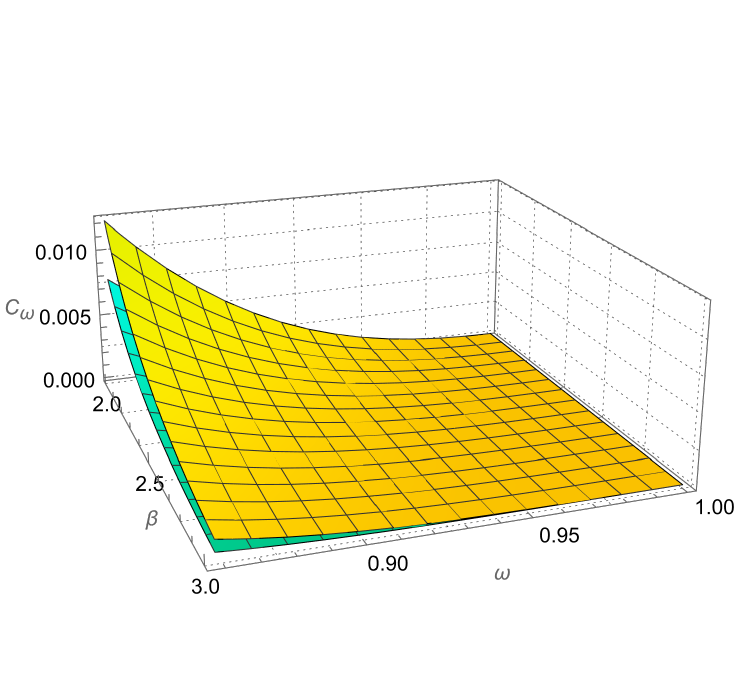}}
\caption{The orange (green) surface represents the result of $W_{1}$ ($W_{2}$).} \label{fig 10}
\end{figure}
\begin{figure}[htbp]
\centering
{\includegraphics[width=6cm,height=3cm]{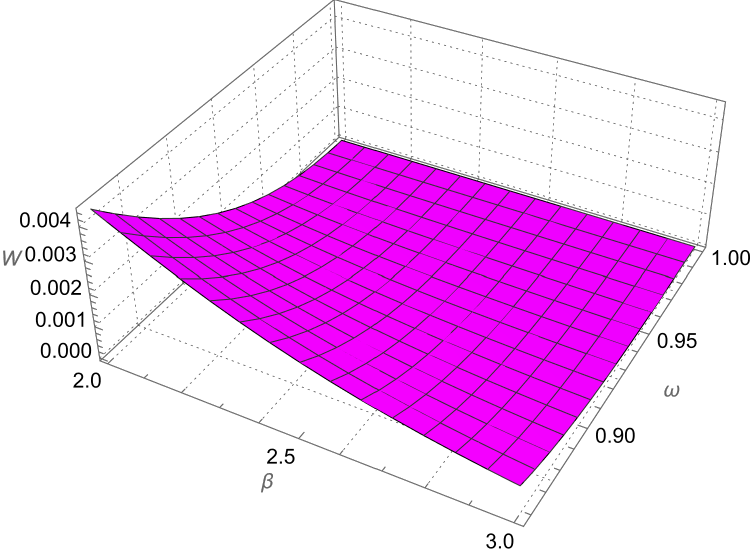}}
\caption{The surface represents $W=W_{1}-W_{2}$.} \label{fig 11}
\end{figure}
\begin{figure}[htbp]
\centering
{\includegraphics[width=6cm,height=4cm]{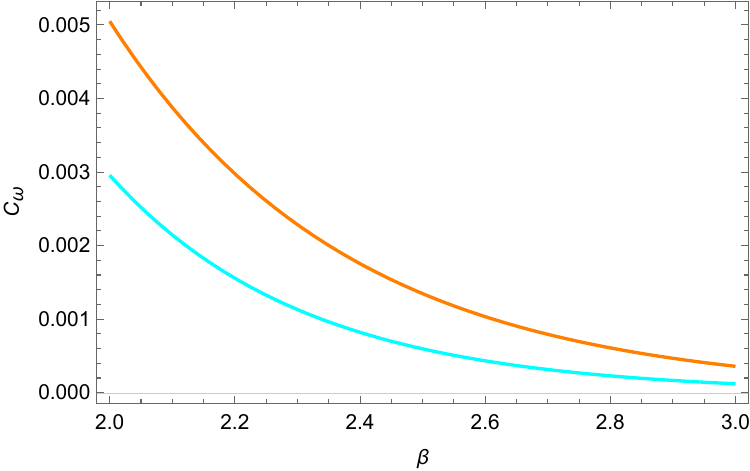}}
\caption{The orange (cyan) line represents our result
of $W_{1}$ $(W_{2})$ with $\omega=0.9$.} \label{fig 12}
\end{figure}

\subsection{Entanglement indicators}
Based on Theorem  \ref{t3}, we introduce a set of multipartite entanglement indicators,
\begin{equation}\label{2.18}
\begin{array}{rl}
\tau_\omega(\rho)=\min\limits_{\{p_l,|\varphi_l\rangle\}}\sum\limits_lp_l\tau_\omega(|\varphi_l\rangle_{A_1|A_2\cdots A_N}),\\
\end{array}
\end{equation}
where the minimum is taken over all feasible ensemble decompositions, $\tau_\omega(|\varphi_l\rangle_{A_1|A_2\cdots A_N})=\mathfrak{C}^2_\omega(|\varphi_l\rangle_{A_1|A_2\cdots A_N})-\sum_{j=2}^{N}\mathfrak{C}^2_\omega(\rho^l_{A_1A_{j}})$ with  $0.85798\leq\omega\leq1$. By using the concavity of the function ${G}_\omega(\rho)$ and following the method in deriving
the squared entanglement of formation \cite{Bai1005032014}, we have the following result.

\begin{theorem}
For any three-qubit state $\rho_{ABC}$, the tripartite entanglement indicator $\tau_\omega(\rho_{ABC})$ is zero for $0.85798\leq\omega\leq1$ iff $\rho_{ABC}$ is biseparable, namely, $\rho_{ABC}=\sum_ip_i\rho_{AB}^i\otimes\rho_{C}^i+\sum_iq_i\rho_{AC}^i\otimes\rho_{B}^i+\sum_ir_i\rho_{A}^i\otimes\rho_{BC}^i$.
\end{theorem}

In fact, the inequality
$\mathfrak{C}^2_\omega(\rho_{A_i|\overline{A_i}})\geq\sum_{j\neq i}\mathfrak{C}^2_\omega(\rho_{A_iA_j})$
is true for any $N$-qubit quantum states. As a consequence, we define a family of indicators,
\begin{equation}\label{2.19}
\begin{array}{rl}
\tau_\omega^i(\rho)=\min\limits_{\{p_l,|\varphi_l\rangle\}}\sum\limits_lp_l\tau_\omega^i(|\varphi_l\rangle_{A_i|\overline{A_i}}),
\end{array}
\end{equation}
where the minimum runs over all feasible pure state decompositions, $\tau_\omega^i(|\varphi_l\rangle_{A_i|\overline{A_i}})=\mathfrak{C}^2_\omega(|\varphi_l\rangle_{A_i|\overline{A_i}})-\sum\limits_{j\neq i}\mathfrak{C}^2_\omega(\rho^l_{A_iA_j})$, $i,j\in\{1,2,\cdots,N-1\}$. $\tau_\omega^i(\rho)$ reduces to  (\ref{2.18}) when $i=1$.

Such a series of entanglement indicators, $\tau_\omega^i(\rho)$, $i=1,2...,n$, allows us to better characterization of the entanglement distribution. They must not be zero for any genuinely multiqubit entangled states. Therefore, these indicators may compensate for the deficiency of tangles. In particular, we evaluate (\ref{2.18}) for W-state and the superposed state of Greenberger-Horne-Zeilinger (GHZ) and W states to demonstrate the universal nature of $\tau_\omega$ as an effective entanglement indicator.

{\noindent\bf Example 3}. For the $N$-qubit $W$ state $|W_n\rangle=\frac{|10\cdots0\rangle+|01\cdots0\rangle+\cdots+|00\cdots1\rangle}{\sqrt{N}}$,
the concurrence between subsystems $A$ and $B_1\cdots B_{N-1}$ is $C(|W_N\rangle)=\frac{2\sqrt{N-1}}{N}$, $C(\rho_{AB_{j-1}})=\frac{2}{N}$, $j=2,\cdots,N-1$. Hence we have
\begin{equation*}
\begin{array}{rl}
\tau_\omega(|W_N\rangle)=&\big[\big(\frac{1+\frac{N-2}{N}}{2}\big)^\omega+\big(\frac{1-\frac{N-2}{N}}{2}\big)^\omega-1\big]^{2\omega}-(N-1)\\
&\big[\big(\frac{1+\frac{\sqrt{N^2-4}}{N}}{2}\big)^\omega+\big(\frac{1-\frac{\sqrt{N^2-4}}{N}}{2}\big)^\omega-1\big]^{2\omega}.\\
\end{array}
\end{equation*}
The three tangle cannot detect the tripartite entanglement of the W-state \cite{Coffman0523062000}. However, the indicator $\tau_\omega$ can efficiently detect the entanglement. Taking $N=3,5,10$, we plot the indicator as a function of $\omega$ in Fig.~\ref{fig 13}. The indicator has non-zero values as long as the state is entangled.
\begin{figure}[htbp]
\centering
{\includegraphics[width=6cm,height=4.5cm]{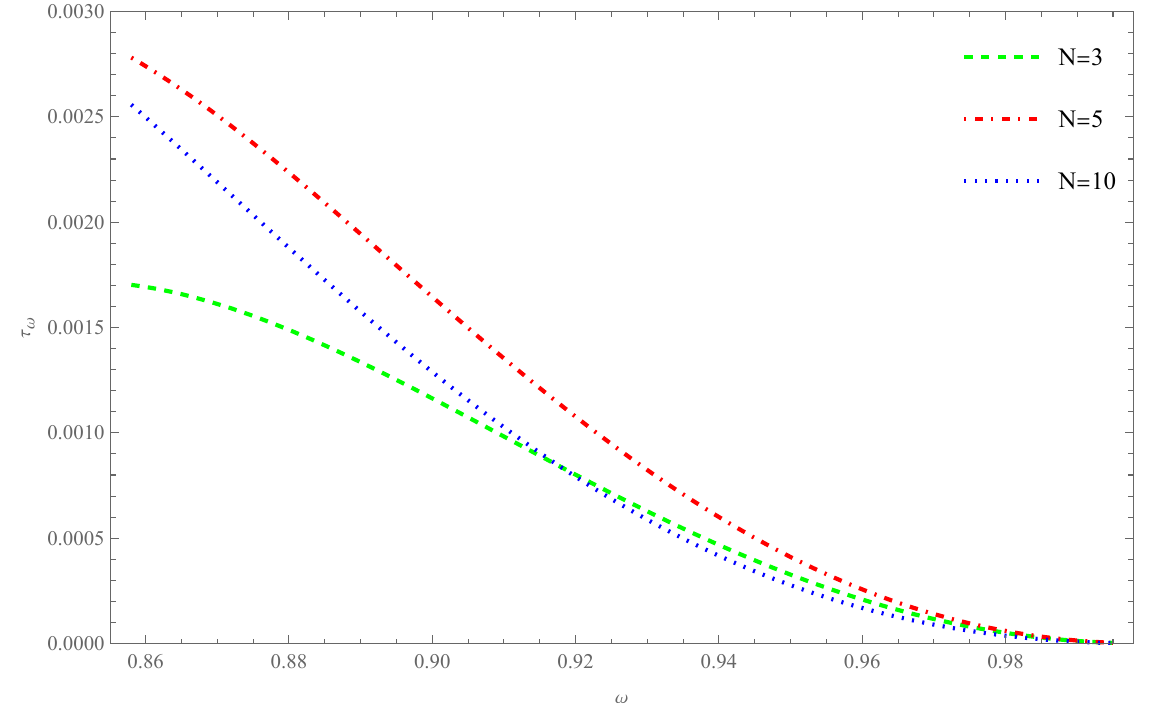}}
\caption{The red thick dotdashed, red thick dotdashed, and blue thick dotted line  correspond to $\tau_\omega(|W_N\rangle)$ with $N=3,5,10$, where $\omega\in[0.85798,1]$.} \label{fig 13}
\end{figure}

{\noindent\bf Example 4}. For the superposed state of the GHZ and the W state,
\be\label{lizi}
|\Psi\rangle_{ABC}=\sqrt{d}|GHZ\rangle-\sqrt{1-d}|W\rangle,
\ee
where $|GHZ\rangle=\frac{1}{\sqrt{2}}(|000\rangle+|111\rangle)$,
$|W\rangle=\frac{1}{\sqrt{3}}(|100\rangle+|010\rangle+|001\rangle)$,
the three tangle of $|\Psi\rangle_{ABC}$ is $\tau_C(|\Psi\rangle_{ABC})=C^{2}(|\psi\rangle_{A|BC})-C^{2}_{AB}-C^{2}_{AC}=\frac{|(9d^2-8\sqrt{6}\sqrt{d(1-d)^3})|}{9}$, which is zero for $d=0$ and $d=0.627$ \cite{Bai1005032014} and does not detect the entanglement. By calculating $\tau_\omega(|\Psi\rangle_{ABC})$ given in (\ref{2.18}) from the analytical formula of $\mathfrak{C}_{\omega}$ for the bipartite states, we show in Fig.\ref{fig 14} that $\tau_\omega$ is positive for all values of $d$.
\begin{figure}[htbp]
\centering
{\includegraphics[width=6cm,height=4.5cm]{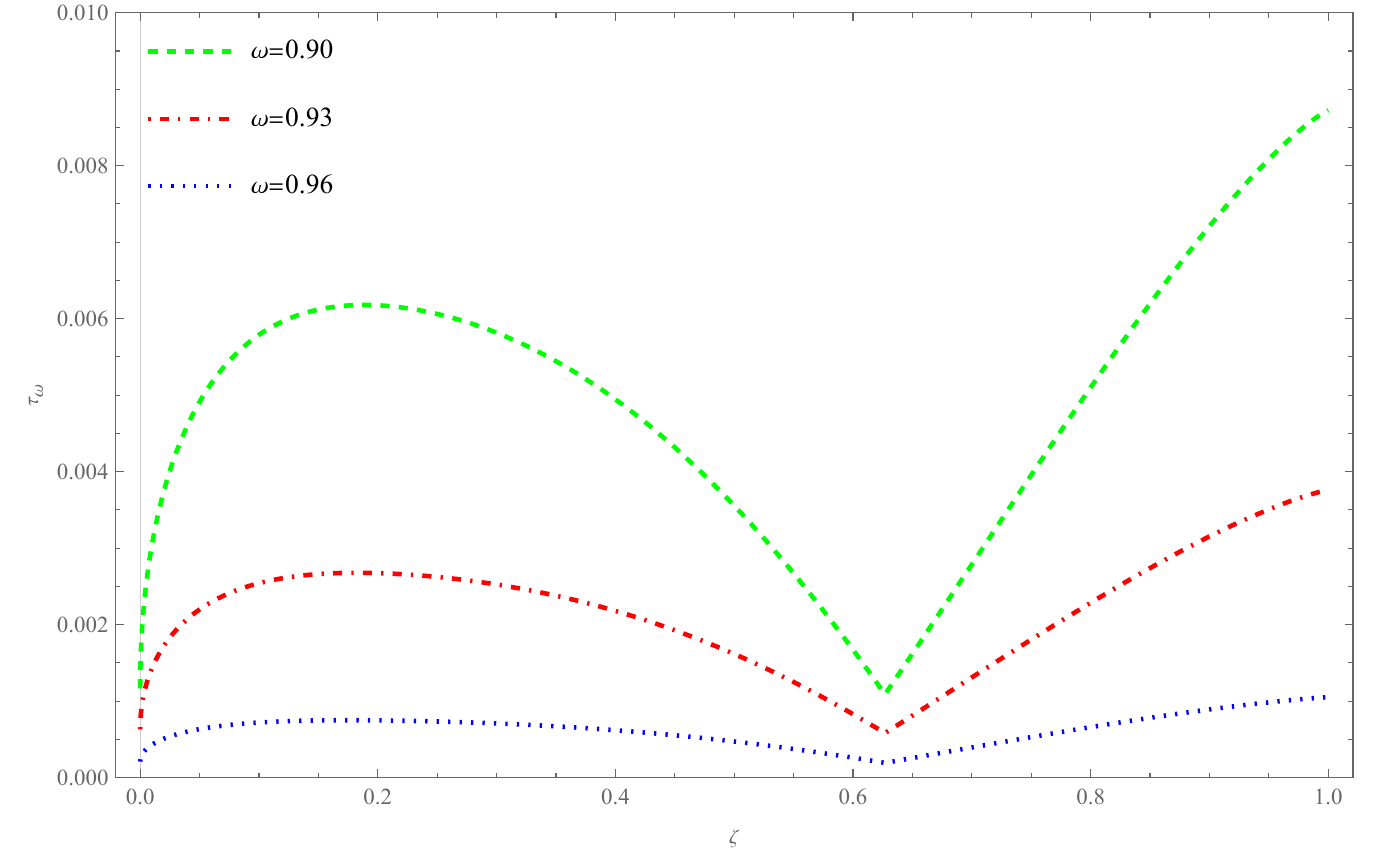}}
\caption{The indicator $\tau_\omega$ of $|\Psi\rangle_{ABC}$ with $\omega=0.90$ (green thick dashed line), $\omega=0.93$ (red thick dotdashed line), and $\omega=0.96$ (blue thick dotted line).} \label{fig 14}
\end{figure}

\subsection{Monogamy properties of squared $G_\omega$C (S$G_\omega$C) and SC}
We analyze the following tripartite pure state of $4\otimes 2\otimes 2$ systems \cite{Bai0623432014},
\begin{equation}\label{mc1}
\ket{\Phi}_{ABC}=\frac{1}{\sqrt{2}}(a\ket{000}+b\ket{110}+a\ket{201}+b\ket{311}),
\end{equation}
where $a=\mbox{cos}\gamma$ and $b=\mbox{sin}\gamma$. The bipartite reduced state for subsystem $AB$ can be written as
\begin{equation}\label{mc2}
\rho_{AB}=\frac{1}{2}\proj{\varphi_1}+\frac{1}{2}\proj{\varphi_2},
\end{equation}
where $\ket{\varphi_1}=a\ket{00}+b\ket{11}$ and $\ket{\varphi_2}=a\ket{20}+b\ket{31}$.
A pure-state in an arbitrary pure state decomposition of $\rho_{AB}$ has the form
\begin{equation}\label{mc3}
\ket{\tilde{\varphi}_i}_{AB}=a_i\ket{\varphi_1}+e^{-i\gamma}\sqrt{1-a_i^2}\ket{\varphi_2}.
\end{equation}
The reduced density matrix is given by $\rho_B^i=\mbox{diag}\{a^2,b^2\}$. Therefore, according to the definitions of Concurrence and $G_{\omega}$C, we have $C^2_{AB}=4a^2b^2$ and $\mathfrak{C}_\omega(\rho_{AB})=(a^{2\omega}+b^{2\omega}-1)^{\omega}$.

Similarly, for the reduced quantum state $\rho_{AC}$, we have $C^2_{AC}=1$ and $\mathfrak{C}_\omega(\rho_{AC})=\big((\frac{1}{2})^{\omega-1}-1\big)^{\omega}$. Moreover, the reduced quantum state of subsystem $A$ is $\rho_{A}=\mbox{diag}\{a^2/2,b^2/2,a^2/2,b^2/2\}$,
from which we get
\begin{equation}\label{mc7}
C^2_{\ket{\Phi}_{A|BC}}=2(1-\mbox{tr}\rho^2_A)=2-a^4-b^4,
\end{equation} and
\begin{equation}\label{mc5}
\mathfrak{C}_\omega(\ket{\Phi}_{A|BC})=\big(2(\frac{a^2}{2})^\omega+2(\frac{b^2}{2})^\omega-1\big)^{2}.
\end{equation}

Thus, the monogamy property of SC is given by
\begin{eqnarray*}
R_{C}(\ket{\Phi}_{A|BC})&\equiv&C^{2}(\ket{\Phi}_{A|BC})-C^{2}(\rho_{AB})-C^{2}(\rho_{AC})\nonumber\\
&=&(2-a^4-b^4)-4a^2b^2-1\nonumber\\
&=&-2a^2b^2,
\end{eqnarray*}
which is polygamous. In a similar way, the monogamy property of  S$G_\omega$C is given by
\begin{eqnarray*}
R_{\omega}(\ket{\Phi}_{A|BC})&\equiv&\mathfrak{C}_\omega^{2}(\ket{\Phi}_{A|BC})-\mathfrak{C}_\omega^{2}(\rho_{AB})-\mathfrak{C}_\omega^{2}(\rho_{AC})\\
&=&\big(2(\frac{a^2}{2})^\omega+2(\frac{b^2}{2})^\omega-1\big)^{2}\\
&\quad~~~-&(a^{2\omega}+b^{2\omega}-1)^{\omega}-\big((\frac{1}{2})^{\omega-1}-1\big)^{\omega}.
\end{eqnarray*}

In Fig. \ref{fig 15}, the parameters are chosen as $a=\mbox{cos}\gamma$ and
$b=\mbox{sin}\gamma$, and the distribution of $R_{\omega}(\ket{\Phi}_{A|BC})$ is plotted as functions of the
parameter $\gamma$, which is non-negative, and therefore $G_\omega$C is monogamous.
\begin{figure}[htbp]
\centering
{\includegraphics[width=6cm,height=4.5cm]{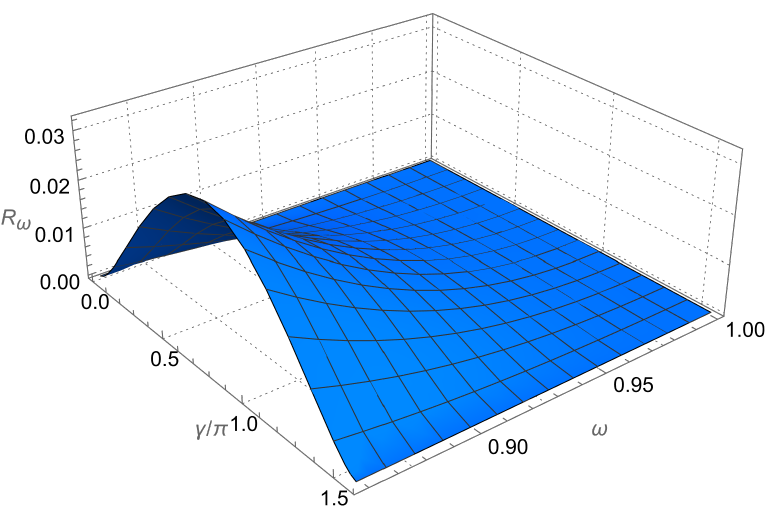}}
\caption{$R_{\omega}(\ket{\Phi}_{A|BC})=\mathfrak{C}_\omega^{2}(\ket{\Phi}_{A|BC})-\mathfrak{C}_\omega^{2}(\rho_{AB})-\mathfrak{C}_\omega^{2}(\rho_{AC})$
as functions of two parameters $\gamma~(0\leq\gamma\leq\frac{\pi}{2})$ and $\omega~(0.85798\leq\omega\leq1)$ in $4\otimes2\otimes2$ systems, where the $SG_{\omega}C$ is monogamous.} \label{fig 15}
\end{figure}

\section{Conclusion}\label{VI}
We have proposed a type of one-parameter entanglement measures, $G_\omega$C ($0<\omega\leq1$), and demonstrate rigorously that they fulfill all the axiomatic requirements of an entanglement measure. In addition, an analytic relation between $G_\omega$C and concurrence has been derived for two-qubit systems. Furthermore, on account of this analytic formula, we have provided the polygamy inequality based on $G_\omega$CoA ($0.85798\leq\omega\leq1$) in multiqubit systems. The squared S $G_\omega$ C has been shown that the squared  S$G_\omega$C obeys a general monogamy relation for arbitrary mixed states of $N$ qubits. Based on the monogamy properties of S$G_\omega$C, we have constructed the corresponding multipartite entanglement indicators, which detect all genuine multiqubit entangled states even for the case that the $N$-tangle fails. Moreover, for multilevel systems, it has been interestingly illustrated that the squared S$G_\omega$C may be monogamous even if the SC is polygamous. These monogamy and polygamy inequalities provide an alternative perspective on the estimation of entanglement of multiqubit quantum states. Our approaches may also be used in future studies aimed at understanding the entanglement distribution in multiqubit systems.

\bigskip
\noindent{\bf Author Contributions:} Writing-original draft, Wen Zhou, Zhong-Xi Shen and Dong-Ping Xuan; Writing-review and editing, Zhi-Xi Wang and Shao-Ming Fei. All authors have read and agreed to the published version of the manuscript.

\bigskip
\noindent{\bf Data availability statement}
All data generated or analyzed during this study are included in this published article.

\bigskip
\noindent{\bf Acknowledgments}
This work is supported by the National Natural Science Foundation of China (NSFC) under Grants 12075159 and 12171044, the specific research fund of the Innovation Platform for Academicians of Hainan Province.


\begin{thebibliography}{99}
\bibitem{Bennett18951993} C. H. Bennett, G. Brassard, C. Cr$\acute{e}$peau, R. Jozsa, A. Peres,
    and W. K. Wootters, Teleporting an Unknown Quantum State via Dual Classical and Einstein-Podolsky-Rosen Channels,
    \href{https://doi.org/10.1103/PhysRevLett.70.1895} {Phys. Rev. Lett. \textbf{70}, 1895 (1993)}.
\bibitem{Bennett28811992} C. H. Bennett and S. J. Wiesner, Communication Via One-and Two-Particle Operators on Einstein-Podolsky-Rosen States,
    \href{https://doi.org/10.1103/PhysRevLett.69.2881} {Phys. Rev. Lett. \textbf{69}, 2881 (1992)}.
\bibitem{Hillery18291999} M. Hillery, V. Bu$\breve{z}$ek, and A. Berthiaume, Quantum secret sharing,
    \href{https://doi.org/10.1103/PhysRevA.59.1829} {Phys. Rev. A \textbf{59}, 1829 (1999)}.
\bibitem{Gisin1452002} N. Gisin, G. Ribordy, W. Tittel, and H. Zbinden, Quantum
    cryptography,
    \href{https://doi.org/10.1103/RevModPhys.74.145} {Phys. Rev. A \textbf{74}, 145 (2002)}.

\bibitem{Vedral16191998} V. Vedral and M. B. Plenio, Entanglement measures and purifi-
    cation procedures,
    \href{https://doi.org/10.1103/PhysRevA.57.1619} {Phys. Rev. A \textbf{57}, 1619 (1998)}.
\bibitem{Vedral22751997} V. Vedral, M. B. Plenio, M. A. Rippin, and P. L. Knight, Quantifying Entanglement,
    \href{https://doi.org/10.1103/PhysRevLett.78.2275} {Phys. Rev. Lett. \textbf{78}, 2275 (1997)}.
\bibitem{Hill50221997} S. Hill and W. K. Wootters, Entanglement of a Pair of Quantum
    Bits,
    \href{https://doi.org/10.1103/PhysRevLett.78.5022} {Phys. Rev. Lett. \textbf{78}, 5022 (1997)}.
\bibitem{Rungta0423152001} P. Rungta, V. Buzek, C. M. Caves, M. Hillery, and G. J. Milburn,
    Universal state inversion and concurrence in arbitrary dimensions,
    \href{https://doi.org/10.1103/PhysRevA.64.042315} {Phys. Rev. A \textbf{64}, 042315 (2001)}.
\bibitem{Wootters22451998} W. K. Wootters, Entanglement of Formation of an Arbitrary
    State of Two Qubits,
    \href{https://doi.org/10.1103/PhysRevLett.80.2245} {Phys. Rev. Lett. \textbf{80}, 2245 (1998)}.
\bibitem{Bennett38241996} C. H. Bennett, D. P. DiVincenzo, J. A. Smolin, and W. K.
    Wootters, Mixed-state entanglement and quantum error correction,
    \href{https://doi.org/10.1103/PhysRevA.54.3824} {Phys. Rev. A \textbf{54}, 3824 (1996)}.
\bibitem{Horodecki32001} M. Horodecki, Entanglement measures,
    \href{https://dl.acm.org/doi/10.5555/2011326.2011328} {Quantum Inf. Comput. \textbf{1}, 3 (2001)}.
\bibitem{Zyczkowski8831998} K. Zyczkowski, P. Horodecki, A. Sanpera, and M. Lewenstein,
    Volume of the set of separable states,
    \href{https://doi.org/10.1103/PhysRevA.58.883} {Phys. Rev. A \textbf{58}, 883 (1998)}.
\bibitem{Vidal0323142002} G. Vidal and R. F. Werner, Computable measure of entanglement,
    \href{https://doi.org/10.1103/PhysRevA.65.032314} {Phys. Rev. A \textbf{65}, 032314 (2002)}.
\bibitem{Kim0623282010} J. S. Kim, Tsallis entropy and entanglement constraints in multiqubit systems,
    \href{https://doi.org/10.1103/PhysRevA.81.062328} {Phys. Rev. A \textbf{81}, 062328 (2010)}.
\bibitem{Gour0121082007} G. Gour, S. Bandyopadhyay, and B. C. Sanders, Dual
    monogamy inequality for entanglement,
    \href{https://doi.org/10.1063/1.2435088} {J. Math. Phys. \textbf{48}, 012108 (2007)}.
\bibitem{Kim4453052010} J. S. Kim and B. C. Sanders, Monogamy of multi-qubit entanglement using R$\acute{e}$nyi entropy,
    \href{https://doi.org/10.1088/1751-8113/43/44/445305} {J. Phys. A \textbf{43}, 445305 (2010)}.

\bibitem{Coffman0523062000} V. Coffman, J. Kundu, and W. K. Wootters, Distributed entanglement,
    \href{https://journals.aps.org/pra/abstract/10.1103/PhysRevA.61.052306} {Phys. Rev. A \textbf{61}, 052306 (2000)}.
\bibitem{Osborne2205032006} T. J. Osborne and F. Verstraete, General monogamy inequality for bipartite qubit entanglement,
    \href{https://journals.aps.org/prl/abstract/10.1103/PhysRevLett.96.220503} {Phys. Rev. Lett. \textbf{96}, 220503 (2006)}.



\bibitem{Oliveira0343032014} T. R. de Oliveira, M. F. Cornelio, and F. F. Fanchini, Monogamy of entanglement of formation,
   \href{https://doi.org/10.1103/10.1103/PhysRevA.89.034303} {Phys. Rev. A \textbf{89}, 034303 (2014).}

\bibitem{Bai1005032014} Y. K. Bai, Y. F. Xu, and Z. D. Wang, General monogamy relation for the entanglement of formation in multiqubit systems,
    \href{https://journals.aps.org/prl/abstract/10.1103/PhysRevLett.113.100503} {Phys. Rev. Lett. \textbf{113}, 100503 (2014)}.

\bibitem{Bai0623432014} Y. K. Bai, Y. F.  Xu, and Z. D. Wang,  Hierarchical monogamy relations for the squared entanglement of formation in multipartite systems,
    \href{https://doi.org/10.1103/10.1103/PhysRevA.90.062343}    {Phys. Rev. A \textbf{90}, 062343 (2014).}

\bibitem{Song0223062016} W. Song, Y. K.  Bai,   M. Yang,  and  Z. L. Cao, General monogamy relation of multi-qubit system in terms of squared R\'{e}nyi-$\alpha$ entanglement,
    \href{https://doi.org/10.1103/PhysRevA.93.022306}  {Phys. Rev. A \textbf{93}, 022306 (2016).}

\bibitem{Luo0623402016} Y. Luo, T. Tian,  L. H. Shao,  and Y. M. Li, General monogamy of Tsallis $q$-entropy entanglement in multiqubit systems,
   \href{https://doi.org/10.1103/PhysRevA.93.062340} {Phys, Rev. A  \textbf{93},  062340  (2016).}

\bibitem{Khan164192019}A. Khan,  J.  ur Rehman,   K. Wang,  and H. Shin, Unified Monogamy Relations of Multipartite Entanglement,
   \href{https://doi.org/10.1038/s41598-019-52817-y}    {Sci.  Rep. \textbf{9}, 16419 (2019).}

\bibitem{Kim0123342018} J. S. Kim, Negativity and tight constraints of multiqubit entanglement,
\href{https://doi.org/10.1103/PhysRevA.97.012334}{Phys. Rev. A \textbf{97}, 012334 (2018).}

\bibitem{Yang5452019} L. M. Yang, B. Chen, S. M. Fei and Z. X. Wang, Tighter constraints of multiqubit entanglement,
\href{https://doi.org/10.1088/0253-6102/71/5/545}{Commun. Theor. Phys. \textbf{71}, 545 (2019).}


\bibitem{Zhu0243042014}X. N. Zhu, and  S. M. Fei, Entanglement monogamy relations of qubit systems,
   \href{https://doi.org/10.1103/PhysRevA.90.024304} {Phys. Rev. A, \textbf{90}, 024304 (2014).}
\bibitem{Seevinck2732010} M. P, Seevinck, Monogamy of correlations versus monogamy of entanglement.
   \href{https://doi.org/10.1007/s11128-009-0161-6} {Quantum Inf. Process. \textbf{9}, 273 (2010).}
\bibitem{Ma3992011} X. S. Ma, B. Dakic, W. Naylor, A. Zeilinger, P. Walther,  Quantum simulation of the wavefunction to probe frustrated Heisenberg spin systems.
   \href{https://doi.org/10.1038/NPHYS1919} {Nat. Phys. \textbf{7}, 399 (2011).}
\bibitem{Ve1072013} E. Verlinde, H. Verlinde, Black hole entanglement and quantum error correction.
   \href{https://doi.org/10.1007/JHEP10(2013)107} {J. High Energy Phys. \textbf{1310}, 107 (2013).}



\bibitem{Horodecki8652009} R. Horodecki, P. Horodecki, M. Horodecki, and K. Horodecki, Quantum entanglement,
    \href{https://journals.aps.org/rmp/abstract/10.1103/RevModPhys.81.865} {Rev. Mod. Phys. \textbf{81}, 865 (2009)}.
\bibitem{Vidal3552000} G. Vidal, Entanglement monotones,
    \href{https://www.tandfonline.com/doi/abs/10.1080/09500340008244048} {J. Mod. Opt. \textbf{47}, 355 (2000)}.
\bibitem{Van0605042013} M. Van den Nest, Universal quantum computation with little entanglement,
    \href{https://journals.aps.org/prl/abstract/10.1103/PhysRevLett.110.060504} {Phys. Rev. Lett. \textbf{110}, 060504 (2013)}.
\bibitem{Mintert2072005} F. Mintert, A. Carvalho, M. Ku\'{s}, A. Buchleitner, Measures and dynamics of entangled states,
    \href{https://www.sciencedirect.com/science/article/pii/S0370157305002334?via} {Phys. Rep. \textbf{415}, 207 (2005)}.
\bibitem{Ando1631989} T. Ando, Majorization, doubly stochastic matrices, and comparison of eigenvalues,
    \href{https://www.sciencedirect.com/science/article/pii/0024379589905806?via} {Linear Algebra Appl. \textbf{118}, 163 (1989)}.
\bibitem{Gour2605012004} G. Gour and B. C. Sanders, Remote preparation and distribution of bipartite entangled states,
    \href{https://journals.aps.org/prl/abstract/10.1103/PhysRevLett.93.260501} {Phys. Rev. Lett. \textbf{93}, 260501 (2004)}.



\end{thebibliography}
\end{document}